\documentclass[pdflatex,sn-mathphys-num]{sn-jnl}% Math and Physical Sciences Numbered Reference Style
\usepackage{graphicx}%
\usepackage{multirow}%
\usepackage{amsmath,amssymb,amsfonts}%
\usepackage{amsthm}%
\usepackage{mathrsfs}%
\usepackage[title]{appendix}%
\usepackage{xcolor}%
\usepackage{textcomp}%
\usepackage{manyfoot}%
\usepackage{booktabs}%
\usepackage{algorithm}%
\usepackage{algorithmicx}%
\usepackage{algpseudocode}%
\usepackage{listings}%
\usepackage{dsfont}
\newcommand{\Id}{\mathbb{I}} % Identity characteristic function
\newcommand{\MoSV}{\mu_{x}} % The measure of spin value 
\newcommand{\MoQP}{\mu_\omega} % The measure of quenched parameter
\newcommand{\QP}{\mathcal{Q}} % The quenched parameter configuration
\newcommand{\R}{\mathcal{R}} % real number 
\newcommand{\SV}{\mathcal{S}} % spin value
\newcommand{\Zd}{{\mathbb{Z}^d}}

\newcommand{\Pb}{{\mathbb{P}}} % Probability over the random fields
\newcommand{\Ex}{\mathbb{E}} % Expectation over the random fields 
\newcommand{\Reg}{\Lambda} % finite region of the lattice 
\newcommand{\ExB}{\partial^{\text{ex}}} % The external boundary of the region 
\newcommand{\InB}{\partial^{\text{in}}} % The internal boundary of the region
\newcommand{\U}{\mathcal{U}} % The unstable region 
\newcommand{\TB}{\Gamma} % The thickened spin boundary set
\newcommand{\C}{\gamma} % The component of the thickened boundary set; The contour 
\newcommand{\sC}{{\bar{\gamma}}} % The support of the contour 
\newcommand{\Ext}{\text{ext}} % external boundary of a set
\newcommand{\Int}{\text{int}} % internal boundary of a set
\newcommand{\CoC}{\mathcal{C}} %The collection of Contours 
\newcommand{\PF}{\mathcal{Z}}
\newcommand{\Ge}{{e_g}} % ground state energy density. 

\newcommand{\Ns}{{N_s}} % number of spin states 
\newcommand{\Ng}{{N_g}} % number of ground states 
\newcommand{\Na}{{N_\alpha}} % number of random fields 
\newcommand{\Nb}{{N_\beta}} % number of quenched parameters
\newcommand{\trsl}{\tau}%translation operator
\newcommand{\FSC}{{\mathcal{F}_c}} % fluctuation-stabilized contour
\newcommand{\QISC}{\mathcal{I}_c} % quasi-invariance-stabilized contour 
\newcommand{\FSIR}{\mathcal{F}_\Int} %fluctuation-stabilized internal region 
\newcommand{\CoB}{\mathscr{B}} % the collection of the cubes/boxes. 
\theoremstyle{thmstyleone}%
\newtheorem{theorem}{Theorem}%  meant for continuous numbers
\newtheorem{proposition}[theorem]{Proposition}% 
\newtheorem{corollary}[theorem]{Corollary}
\newtheorem{lemma}[theorem]{Lemma} 
\theoremstyle{thmstyletwo}%
\newtheorem{remark}{Remark}%

\theoremstyle{thmstylethree}%
\newcounter{def}
\newtheorem{definition}[def]{Definition}%

\begin{document}

\title{The stability of long-range order in disordered systems: A generalized Ding-Zhuang argument}

\author*[1,2]{Yejia Chen}\email{chenyejia@itp.ac.cn}

\author*[1,2]{Jianwen Zhou}\email{zhoujianwen@itp.ac.cn}

\author[1,2]{Ruifeng Liu}\email{liuruifeng@itp.ac.cn}

\author*[1,2,3]{Hai-Jun Zhou}\email{zhouhj@itp.ac.cn}

\affil*[1]{
  Institute of Theoretical Physics, Chinese Academy of Sciences, Zhong-Guan-Cun East Road 55, Beijing 100190, China
}

\affil[2]{School of Physical Sciences, University of Chinese Academy of Sciences, Beijing 100049, China
}

\affil[3]{MinJiang Collaborative Center for Theoretical Physics, MinJiang University, Fuzhou 350108, China
}

\abstract{
The stability of long-range order against quenched disorder is a central problem in statistical mechanics. This paper develops a generalized framework extending the Ding-Zhuang method and integrated with the Pirogov-Sinai framework, establishing a systematic scheme for studying phase transitions of long-range order in disordered systems. We axiomatize the Ding-Zhuang approach into a theoretical framework consisting of the Peierls condition and a local symmetry condition. For systems in dimensions $d \geq 3$ satisfying these conditions, we prove the persistence of long-range order at low temperatures and under weak disorder, with multiple coexisting distinct Gibbs states. The framework's versatility is demonstrated for diverse models, providing a systematic extension of Peierls methods to disordered systems.
}




\maketitle

\section{Introduction}\label{sec1}

Phase transitions represent singularities in the thermodynamic behavior of many-body systems and lie at the heart of equilibrium statistical mechanics. The phenomenon of long-range order, characterized by the spontaneous breaking of symmetry in the classical Landau's picture, has been rigorously established for various ordered systems, especially for lattice systems. However, the persistence of such ordered phases under quenched disorder, where randomness is frozen in time, remains a profound theoretical challenge. This article addresses the stability of long-range order in disordered lattice systems, where impurities, random fields, or heterogeneous couplings might disrupt collective order. The central question we confront is: Under what conditions can robust phase transitions survive the destabilizing influence of quenched disorder?

The foundational work of Peierls \cite{Peierls_1936} provided the first rigorous demonstration of long-range order in the two-dimensional Ising model. His argument rested on a geometric insight: the energetic penalty for creating a domain of flipped spins scales with its boundary length, while the entropic contribution scales with the number of possible domain configurations. In dimensions $d \geq 2$, the energy cost dominates over the entropy gain, stabilizing ordered phases. This energy-entropy balance framework became paradigmatic and was extended in many ways, including the extensions to systems with continuous symmetries  \cite{malyshevphase1975, VANBEIJEREN1978145}, quantum systems \cite{robinsonproof1969}, and systems with long-range interactions \cite{affonsolongrange2024}.  

Of all these advances, the Pirogov-Sinai theory \cite{Pirogov-SinaiI, Pirogov-SinaiII} revolutionized the landscape of the Peierls argument. By introducing a systematic low-temperature expansion via contours, the excitation domains separating distinct ground states, and the polymers consisting of contours, Pirogov and Sinai quantified how local ground-state properties extend to finite temperatures. Crucially, this framework established that multiple coexisting phases emerge from distinct ground-state configurations precisely when the system admits local symmetries compatible with the contour dynamics, which are represented by transformations that preserve the Hamiltonian while permuting the ground states. Furthermore, Zahradník \cite{Pirogov-Sinai-Zahradnik} contributed substantially to the theory by introducing fundamentally new ideas about truncated pressures so that it can describe the phase diagram even without the symmetry conditions. This framework can even be used in an algorithmic way \cite{Algorithmic_PST}. There are also methods developed differently from the Peierls-type argument, such as the Infrared-bound method \cite{Frohlich-InfraredBoundsPhase-1976} and transfer matrix method \cite{baxterHardsquareLatticeGas}, some of which can even offer quantitative results for interesting physical quantities, especially for two-dimensional systems \cite{sherman1960,baxterExact1980, Baxter-PottsModelCritical-1973, Duminil-Copin-DiscontinuityPhaseTransition-2016,Duminil-Copin-ContinuityPhaseTransition-2017,Kac_Ward1952}.

The introduction of quenched disorder fundamentally disrupts this picture, however. Random fields break symmetry invariance almost everywhere in the lattice and suppress order parameters, necessitating novel adaptations of these symmetry principles. Among the disordered systems, one of the paradigmatic examples is the Random Field Ising Model (RFIM) with Hamiltonian
\begin{equation}
\mathcal{H} \, = \,  -J \sum_{\langle i j \rangle} \sigma_i \sigma_j - \sum_i \eta_i \sigma_i, \quad \quad \sigma_i = \pm 1
\end{equation}
where $\langle i j \rangle$ means a pair of nearest neighbors and $J>0$ is the ferromagnetic coupling constant, and the fields $\eta_i$ are typically independent identically distributed Gaussian variables with zero mean and variance $\epsilon^2$. 

One of the first heuristic arguments for the RFIM was introduced by Imry and Ma \cite{Imry_Ma}, demonstrating that in dimensions $d \leq 2$, an arbitrarily weak random field destroys long-range order. In contrast, for $d > 2$, long-range order can persist under weak disorder, but it is destroyed under strong disorder where the energy gain prevails. Several works employing heuristic methods, numerical simulations, and experimental protocols also predicted this dimensional dichotomy in other systems with quenched disorder, including diluted antiferromagnets in uniform external fields \cite{SFishman-RandomFieldEffects-1979, diluteAntiferromagnet}, Migdal-Kadanoff renormalization-group scheme \cite{Cao-Migdal-Kadanoff-study-RFIM-1993}, connections with Anderson-Mott transitions \cite{AndersonMottTransition}, numerical simulation results of the three-dimensional (3D) RFIM \cite{Fytas-Universality-3dRFIM-2013, Theodorakis-T0-RFIM-2014} and diluted antiferromagnets \cite{Ahrens-2013}, experiments of Ising Jahn-Teller transitions \cite{Jahn-Teller}, binary liquids in porous media \cite{Binaryliquid} and hydrogen in metal alloys \cite{Hydrogen}.

Rigorous theoretical progress has also been made on the dimensional dichotomy in disordered systems.
Aizenman-Wehr argument \cite{aizenmanRoundingEffectsQuenched1990} established a fundamental ``no-go" result for general disordered systems: systems with either continuous symmetry ($d \leq 4$) or discrete symmetry ($d \leq 2$) cannot exhibit spontaneous symmetry breaking under quenched disorder of arbitrarily weak random fields for some specific order parameters, including the case of RFIM. For the existence of long-range order in the RFIM, Imbrie \cite{imbrieground1985} and Bricmont and Kupiainen \cite{Bricmont-PhaseTransition3d-1988} developed intricate methods for dimensions $d \ge 3$ to show that long-range order persists when the quenched parameters are sufficiently small. 

Recently, Ding and Zhuang offered a simplified proof using an extended Peierls method \cite{Ding-Zhuang}, inspiring the generalization of the present work. Based on the multiscale analysis given in \cite{FFS84}, the Ding-Zhuang argument provides a powerful and rigorous framework for establishing phase coexistence in disordered systems via symmetry operations over both spin and random field configurations. This approach is quite efficient and has been extended to other systems like long-range interacting RFIM \cite{Maia}, and to other results over RFIM such as the description of the phase diagram \cite{Ding-LongRangeOrder-2024} and the critical behavior with respect to the size of the region \cite{Ding-PhaseTransitionCritical-2023}. 

Despite these significant progresses, establishing a unified framework for systematically investigating long-range order, whether discrete or continuous,  in disordered lattice systems has still remained challenging. This work bridges this gap by generalizing the Ding-Zhuang methodology, unifying and extending recent advances in the qualitative analysis of long-range ordered phases. We consider general statistical systems on lattices with weakly dependent inter-site interactions and weak random fields. Within the Pirogov-Sinai framework, we characterize contour energy bounds through Peierls conditions, while categorizing random field effects into three fundamental types (Section \ref{sec:Ding-Zhuang-general}).

Our generalization further addresses diverse local symmetries inherent to disordered systems. Notably, in systems without quenched disorder, Peierls arguments utilize not only spin reflection symmetry but also translation symmetry \cite{dobrushin1969_2, heilmannUseReflectionSymmetry1974} and other lattice-dependent symmetries \cite{heilmannUseReflectionSymmetry1974, heilmannexistence1972}. We axiomatize these local symmetries in Definitions \ref{def:local-symmetry} and \ref{def:extended-local-symmetry}, proving that for disordered systems in $d \geq 3$, the coexistence of Peierls conditions in unperturbed systems and local symmetry operations guarantees the persistence of long-range order at low temperatures under weak random fields.

This paper is organized as follows. In Section \ref{sec:notation}, we introduce the model and our notation. Section \ref{sec:contour} develops the Pirogov-Sinai contour representation and its corresponding polymer model. The core of our proof is presented in Section \ref{sec:Ding-zhuang}, where we establish probabilistic control over relevant physical quantities using concentration inequalities and a multiscale chaining argument. With these tools, we prove our main theorems in Section \ref{sec:theorems} and demonstrate their application to a variety of models in Section \ref{sec-model-examples}.

\section{Notations on statistical models}\label{sec:notation}

In this section, we set up the necessary notations and definitions of our argument. In all cases, the statistical models are located on the $d$-dimensional lattice $\Zd$ with origin denoted by $0$ without ambiguity. In the following, the distance on $\Zd$ is defined by the $L^\infty$ norm. For a vector $s \in \Zd$, its $L^\infty$ norm is simply the maximum absolute value of its entries. The state space of the statistical mechanical system is called the space of spin configurations, denoted by $x = \{x_s\} \in \SV^\Zd$, where $\SV$ is a compact metric space (the spin value space) and $x_s \in \SV$ denotes the spin value of site $s \in \Zd$. Unless otherwise specified, $\SV$ is assumed to be a finite discrete set with $\Ns := |\SV| < \infty$; this assumption will be relaxed in Section \ref{sec:theorems} for continuous spin models.

We are going to construct a real-valued Hamiltonian defined over spin configurations $x\in \SV^\Zd$ for each model, excluding the presence of random fields. This Hamiltonian comprises periodic arrays of local interaction functions $g^{\alpha}(x)$, where the index $\alpha$ distinguishes different types of interactions. Each interaction function $g^{\alpha}(x)$ depends solely on the spins within a finite range around the origin $0$, characterized by a constant $R_\alpha>0$. Specifically, $g^{\alpha}(x)$ is a function of the spin variables $x_s$ for all $s$ satisfying $s\leq |R_\alpha|$. Furthermore, each interaction function is bounded, such that $|g^{\alpha}(x)|\le 1$.
To incorporate spatial homogeneity, we define the translated interaction function $g^{\alpha}_s(x)$ by $g^{\alpha}_s(\trsl_s x)=g^{\alpha}(x)$, where $\trsl_s$ denotes the translation operator that shifts the spin configuration by a vector $s \in \Zd$. In other words, $(\trsl_s x)_t= x_{s+t}$.

To perform the Peierls-type argument, we must compare the conditional probabilities between different configurations, which leads to the definition of the Hamiltonian in a finite region $\Reg \subset \Zd $. To investigate the problem in a finite region, we must design the boundary conditions as a transformation in the space of spin configurations $\SV^{\Zd}$. For any $x, b \in \SV^{\Zd}$, we denote $x^{b,\Reg}$ through
\begin{equation}
\label{eq:boundary-transform}
    x^{b, \Reg}_{s} = \left\{ \begin{array}{cc}
        x_s  &  s \in \Reg \\
        b_s  &  s \notin \Reg
    \end{array} \right..
\end{equation}
Here, the auxiliary spin configuration $b$ serves as the boundary condition, fixing the spin values outside $\Reg$ for any configuration $x$. With the above notations, we define the local Hamiltonian $H^{b}_{0, \Reg}(x)$ over $\Reg$ without boundary as: 
\begin{equation}
\label{eq:original_Hamiltonian}
H_{0, \Reg}(x)\, = \, \sum_{\alpha} \sum_{s \in \Reg } h^\alpha g^\alpha_s(x) ,
\end{equation}
where $h^\alpha \in \mathbb{R}$ stands for the constant strength of the interaction for each type, and the local Hamiltonian with the boundary condition $b$ is $H^b_{0, \Reg}(x) = H_{0,\Reg}(x^{b, \Reg})$.

For further usage, we define the external boundary set and the internal boundary set of the finite region $\Reg$ as
\begin{eqnarray}
\ExB_n \Reg   & = &   \{ s \in \Zd |  s \notin \Reg, d(s, \Reg) = n \} ,
\label{eq:internal-boundary}
\\
\InB_n \Reg & = & \{ s \in \Zd |  s \in \Reg, d(s, \Reg^c) = n \},
\label{eq:external-boundary}
\end{eqnarray}
where $d(i, \Reg)$ stands for the distance induced by the $L^\infty$ norm and $\Reg^c$ represents the sites of $\Zd$ not in $\Reg$. When there is no ambiguity, we will write $\ExB \Reg = \ExB_1 \Reg$ and $\InB \Reg = \InB_1 \Reg$ for convenience.  

Now, we consider the case with additional random fields interacting with the system. All random fields are generated by the following rules:
\begin{itemize}
    \item There is an array of independent and identically distributed random variables $\omega=\{\omega^{\beta}_s\}$ at every site $s \in \Zd$ with $\omega_s^\beta \in \QP$, where $\QP$ is a probability space representing the quenched parameter space, common for all $\beta$ without loss of generality, and $\Nb$ is the index set for $\beta$. These random variables represent physical quantities directly coupling to the thermalized environment before quenching. In the following, we refer to $\omega$ as the quenched parameter configuration. 
    
    \item There is another array of random variables $\eta = \{ \eta^{\alpha}_s\} \in \R^{\Na \times \Zd}$, where $\Na$ is the index set of $\alpha$, and each $\eta^{\alpha}_s$ is a local continuous function of $\omega$. To enforce locality, we fix  constants $D_\alpha>0$ and subsets $\Nb(\alpha) \subseteq \Nb$ such that $\eta^{\alpha}_s$ depends on $\omega^{\beta}_t$ only for $|t-s|\le D_\alpha$ and $\beta \in \Nb(\alpha)$. Furthermore, we also assume the translation covariance of $\eta^{\alpha}$, or $\eta^{\alpha}_{s+t}(\trsl_t \omega) = \eta^{\alpha}_s(\omega)$ in other words. In the following, we refer to $\eta$ as the random field configuration. 
    \item There exists local real interaction functions $g'^{\alpha}(x)$ of the spin configuration $x$ which share similar properties as $g^{\alpha}(x)$. Without ambiguity, we denote them as $g^{\alpha}(x)$ also.
\end{itemize}

Combining the periodic Hamiltonian $H^b_{0,\Reg}$ in \eqref{eq:original_Hamiltonian} with the coupling to the additional random fields $\eta$, we obtain the Hamiltonian of a general disordered statistical mechanical system:
\begin{equation}\label{eq:general_Hamiltonian}
H^{b}_{\eta, \Reg}(x) = H^{b}_{0,\Reg}(x) + \sum_{\alpha} \sum_{s \in \Reg } \eta^\alpha_s g^{\alpha}_s(x^{b, \Reg}). 
\end{equation} 

By Pirogov-Sinai theory, for the disorder-free system \eqref{eq:original_Hamiltonian}, long-range behavior in the low-temperature regime arises from the stability of ground states under thermal fluctuations. A spin configuration $x$ is a local ground state if for any finite $\Reg \subset \Zd$ and any $x'$ coinciding with $x$ outside $\Reg$, $H_{0,\Reg'}(x) \leq H_{0, \Reg'}(x')$ for any subset $\Reg' \supseteq \Reg \cup \ExB_{ 2 \max_\alpha\{ R_\alpha \}} \Reg $. In this article, we focus on the periodic local ground state. A local ground state $x$ is also a periodic local ground state if there exists a basis $\{v_i\}_{i=1}^d$ of $\Zd$ such that $x_{s+v_i} = x_s$ for all $s$. We denote $e_i$ as the $i$-th unit axis vector, and $P_b$ as the least common multiple of a series of positive integer $a_i$ such that for every $i$, $a_i e_i$ is the integer combination of $v_i$. Define the characteristic range $R := (\max_\alpha\{R_\alpha\}+\max_\alpha\{D_\alpha\})P_b$. To simplify the problem, one can choose $R = 1$ without loss of generality by modifying the configuration spaces. More explicitly, we can set $\SV' = \SV^{R^d}$ to be the new value space consisting of the spin values within a hypercubic space. It is obvious that there is a natural one-to-one correspondence between $\SV'^{\Zd}$ and $\SV^{\Zd}$ and we can also rewrite all the contributing terms in the new spin configuration space with $R = 1$. Moreover, we also find that the ground state is the constant function over $\Zd$. Therefore, we use $b^k \in \SV$ to label not only the spin values but also the boundary configurations of these ground states. The ground state energy per site is $\Ge := \sum_\alpha h^\alpha g^\alpha_s(b^k)$ for any $s \in \Zd$ and any type of boundary condition $B^k$ by the definition of ground states and the translation invariance. The number of ground states is also denoted by $\Ng$.   

\begin{remark} 
It should be warned that the pair of the same type interaction term acting at yet different sites $g^\alpha_{s_1}(x)$ and $g^\alpha_{s_2}(x)$ might represent different type of interaction in the modified spin configuration space. At the same time, apparent changes should also be made for the dependence of the new types of interacting functions of the random field over the types of the quenched parameters. 
\end{remark}

Proving phase transitions in \eqref{eq:general_Hamiltonian} also requires quantifying the stability of long-range order in \eqref{eq:original_Hamiltonian} under the random fields, which leads to the assumption that $\eta$ is weak in the sense that $\eta^\alpha_s$ is concentrated near $0$. In the following, let $(\MoQP^\beta)_s$ be the distribution of $\omega_s^\beta$ and the full measure is $\bigotimes_{\Nb \times \Zd} (\MoQP^\beta)_s$ where $\bigotimes$ denotes the product measure. 

Also, assume $\QP \subset \mathbb{R}$ and each $\eta^\alpha_s$ is a Lipschitz continuous function over $\omega$ such that $|\nabla \eta^\alpha_s|(\omega) \le (3D_\alpha)^{-d}$ so that $H^b_{\eta,\Reg}$ can be a Lipschitz continuous function with respect to any $\omega^\beta_s$ with constant $1$. We also assume that
\begin{equation}
\label{eq:quenched-condition}
\Ex(\eta_s^\alpha) = 0, \,
 \quad  \eta_s (\omega \equiv 0)= 0, \, \quad \epsilon_\beta :=  \sqrt{\Ex[(\omega_s^\beta)^2]}
\end{equation}
without loss of generality. By Lemma \ref{lem:variation-boundedness} and \cite[Theorem 2.3]{APC550}, $\Ex[ (\eta^\alpha_s)^2]$ is bounded by $\epsilon_{\alpha}^2 := \sum_{\beta \in \Nb(\alpha)} \epsilon_\beta^2$. Let the maximal standard deviation $\epsilon := \sqrt{\sum_\alpha (\epsilon_\alpha)^2}$ represents the extent of the weakness of the random fields. Henceforth, we assume all random fields ${\omega_s^\beta}$ belong exclusively to one of these two classes: Gaussian random variables, or Bounded random variables, and satisfy \eqref{eq:quenched-condition}. This restriction streamlines our study of phase transitions under quenched disorder while preserving physical generality. In the bounded case, we also assume $|\omega_s^\beta| \le 1$. 

\section{Contours and polymer models}\label{sec:contour}

One of the best ways to elucidate the fluctuation properties of the Gibbs distribution is to depict the spin configuration via contours. In this section, we define contours and formulate the quantities using a corresponding polymer model. 

\subsection{Concepts over contours}\label{sec:contour-concept}

To investigate the general Gibbs distribution of the system, we need to find a proper representation of the excitation state based on the ground state. The definition of contours can be a powerful tool in this context. For a spin configuration $x$, a site $s \in \Zd$ is called unstable if for every ground state spin value $b^k$, there exists $t \in \Zd$ with $|s-t| \le 1$ such that $x_t \neq b^k$. Define $\U(x) \subseteq \Zd$ as the unstable region of $x$, consisting of all the unstable sites. With the above definition, we introduce the thickened unstable region:
\begin{equation} \label{eq:Tcontour}
    \TB(x) := \{s \in \Zd| d(s, \U(x)) \le 1 \} \; .
\end{equation}

To facilitate computation of the Hamiltonian and associated physical quantities for spin configurations, we partition the regions such as $\Reg$ and $\TB(x)$ into subdomains for localized analysis. This requires defining restricted configurations, a framework that enables rigorous characterization of contours. Contours decompose the thickened unstable set into connected components based on their underlying configurations. We now formally define restricted spin configurations and contours.

\begin{definition}[Restricted spin configurations and contours]
\label{def:r-config&contour}
A restricted spin configuration $r$ is defined as a pair
\begin{equation}
    r = (\bar{r}, \{r_s\}),
\end{equation}
where:
\begin{itemize}
    \item $\bar{r} \subseteq  \Zd$, called the support of $r$, is the definition domain of the restricted configuration.
    \item $r_s$ is the spin value of the restricted configuration at site $s \in \bar{r}$.  
\end{itemize}
A contour $\C=(\sC,\{\C_s\})$ is a restricted spin configuration with the additional requirements: %
\begin{itemize}
    \item $\sC$ is a path-connected subset of $\Zd$ in the sense that for any pair $s$, $s^\prime \in \sC$, there exists a sequence $\{s_i\}$ such that  $|s_i-s_{i-1}|=1$ and the two end points equal to $s$ and $s^\prime$.
 
    \item $\{\C_s\}$ is the restricted spin configuration that satisfies: (i) In every connected component of $\InB_2 \sC$, $\C_s \equiv b^k$ for some ground state spin value $b^k$. (ii) Any site $s \in \sC \setminus \InB_2 \sC$ is unstable with respect to $\C$.      
\end{itemize}
\end{definition}

In the context of Definition \ref{def:r-config&contour}, the restriction of $x$ to $\TB(x)$ decomposes into contours $\{ \C_i(x) \}$. In particular, the contour conditions for $(\C_i(x))_s$ can be easily verified by the geometric observation. In the following, we also denote by $\TB(x):=\{\C_i(x)\}$ as the collection of the restricted spin configurations. 

To classify the contours in a suitable way, it is useful to consider the connected components of $\sC^c = \cup_i A_i$, which are finite for a contour of finite support $\sC$. It is also obvious in geometry that $\cup_i \ExB A_i = \InB \sC$. Moreover, each $\ExB A_i$ is connected, but $\ExB A_i \cup \ExB A_j$ is disconnected for any pair $i,j$ with $i \neq j$. For any finite contour $\sC$,  there exists a unique infinite connected component in $\sC^c$, denoted as $A_0$ after reordering. We refer to this component as $\Ext \C$. Additionally, we denote the union of the other connected components as $\Int \C$. For each component $A_i \subset \Int \sC$, we assign $A_i \subset \Int_k \sC \subset \Int \sC$ if $\C_s = b^k$ for all $s \in \ExB A_i$.  Consequently, it follows that $\cup_k \Int_k \sC = \Int \sC$. The contour $\C$  can also be categorized by $\ExB \Ext \C$, setting $\CoC^k$ the collection of the contours such that for any $\C \in \CoC^k$, $\C_s = b^k$ for all $ s\in \ExB \Ext \C$.

\subsection{Rewriting as polymer model}\label{sec:polymer}

The framework of restricted spin configurations and contours enables effective decomposition of the Hamiltonian. For example, we define $H_{\eta,\Reg}(r)$ for the restricted spin configuration $r$ with $\bar{r} \supseteq \Reg \cup \ExB \Reg$ due to the locality of the interaction function. We can also define $H_{\eta}(\C)$ for a contour $\C$ as the Hamiltonian over $\sC$, since the boundary condition over $\ExB \sC$ can be inferred from the values of $\C$ over $\InB \sC$ by the contour condition. In the following, we express the partition function of the general Hamiltonian in \eqref{eq:general_Hamiltonian} as a polymer model with every atom representing a contour. We define $H^k_{\eta,\Reg}:= H^{b^k}_{\eta,\Reg}$ for simplicity. The partition function with boundary condition $b^k$ of restricted spin configurations $x$ in $\Reg$, whose boundary set is also contained in $\Reg$, is defined as: 
\begin{equation}\label{eq:partition-function}
    \PF^k_{\eta,\Reg} := \sum_{x} \exp\left[- \frac{1}{T} H^k_{\eta, \Reg }(x)\right],
\end{equation}
where the sum ranges over all possible configurations $x$ restricted in $\Reg$ with $x_s = b^k$ for $s \in \InB_2 \Reg$ and $T$ represents the temperature. For further usage, we also define a modified partition function over the spin configurations $x$ whose boundary set satisfies $\TB(x) \subset \Reg$ and $d(\TB(x),\Reg^c)>1$ denoted by $\tilde{\PF}^k_{\eta, \Reg}$. Specifically, $\tilde{\PF}^k_{\eta, \Reg}$ equals \eqref{eq:partition-function} but sums over all possible configurations of $x$ restricted in $\Reg$ with $x_s = b^k$ for $s \in \InB_3 \Reg$. Accordingly, the Gibbs measure is defined by its density over the configuration $x$ restricted in $\Reg$ with $x_s = b^k$ for $s \in \InB_2 \Reg$ as
\begin{equation}\label{eq:Gibbs-measure}
\mu^k_{\eta, \Reg}(x) := \frac{\exp\left[- \frac{1}{T} H^k_{\eta, \Reg}(x)\right] }{\PF^k_{\eta, \Reg}}.
\end{equation}

An inductive approach provides a systematic approach for computing the partition function and related quantities. Let $\TB$ be a collection of contours. We can define a partial order on the contours of $\TB$ such that $\C \le \C'$ if $\sC \subset \Int \C'$. Under this ordering, an external contour is characterized as a maximal element. Thus we can extract a subcollection $\TB_\Ext \subseteq \TB$ consisting of all the external contours in $\TB$. We also denote by $\TB_{\Ext}(x)$ the collections of all the external contours in $\TB(x)$. We define the external region for a contour collection $\TB$ as the intersection of the external regions of all the contours. Specifically,  
\begin{equation}\label{eq:def-Ext}
\Ext(\TB):= \bigcap_{\C \in \TB} \Ext(\C) = \bigcap_{\C \in \TB_{\Ext}} \Ext(\C). 
\end{equation}
For configuration $x$, set also $\Ext(x):= \Ext(\TB_{\Ext}(x))$.

By extracting the external contours and the external area, we can decompose the Hamiltonian as follows:
\begin{equation}\label{eq:Hamiltonian-decomposition}
    H^{k_0}_{\eta, \Reg}(x) = H^{k_0}_{\eta, \Reg \cap \Ext(x) }(b^{k_0}) + \sum_{\C \in \TB_{\Ext}(x)} \left\{ H_{\eta} (\C)  + \sum_k H^k_{\eta, \Int_k \C  }(x) \right\}
\end{equation}

In the following, we denote by $\TB_{\Ext}^{k} \subseteq \CoC^k$ an admissible collection of external contours in $\CoC^{k}$; a collection of contours is called admissible if any two distinct contours $\C$ and $\C'$ satisfy $d(\sC, \sC')>1$. Therefore, from the decomposition \eqref{eq:Hamiltonian-decomposition}, the corresponding partition function can be derived by summing up the cases of all the possible $\TB^{k_0}_{\Ext}$ as
\begin{equation}\label{eq:PF-decomposition}
\begin{aligned}
\PF^{k_0}_{\eta, \Reg}  = & \sum_{\TB^{k_0}_\Ext} \exp\left(-\frac{\Ge}{T} |\Ext(\TB^{k_0}_\Ext) \cap \Reg| \right) \exp\left[-\frac{1}{T}\sum_{s \in \Ext(\TB^{k_0}_\Ext) \cap \Reg} \sum_{\alpha} \eta^\alpha_s g^\alpha_s(b^{k_0}) \right]  \\
 & \prod_{\C \in \TB^{k_0}_\Ext}\left\{  \exp \left[-\frac{1}{T} H_{\eta}(\C) \right]  \prod_k \tilde{\PF}^k_{\eta, \Int_k \C } \right\}.   
\end{aligned}
\end{equation}

Comparing each term in (\ref{eq:PF-decomposition}) with the ground state term $\exp(-|\Reg|\Ge)$, we derive that
\begin{equation}\label{eq:recursion-pre}
\begin{aligned}
&\exp\left[\frac{\Ge|\Reg|+ S^{k_0}_\Reg}{T}\right] \PF^{k_0}_{\eta,\Reg} \\
= &  \sum_{\TB^{k_0}_\Ext} \prod_{\C \in \TB^{k_0}_\Ext} \left\{\exp\left[-\frac{1}{T}D^{k_0}_\eta(\C) \right] \prod_k \exp\left[\frac{\Ge|\Int_k \C|+ S^{k_0}_{\Int_k \C} }{T} \right]\tilde{\PF}^k_{\eta, \Int_k \C} \right\} 
\end{aligned}
\end{equation}
where $S^k_\Reg =\sum_{s \in 
\Reg} \sum_\alpha \eta^\alpha_s g^\alpha_s (b^k)$ is a random variable over the finite region $\Reg$ and $D^{k^0}_\eta (\C)$ for $\C \in \CoC^{k_0}$ stands for the difference of the Hamiltonian between the excited state and the ground state:  
\begin{equation}\label{eq:difference-formula}
    D^k_\eta(\C):= H_{\eta}(\C) -H^k_{\eta, \sC}(b^k) =  \sum_{ s \in \sC } \sum_\alpha  \left[(h^\alpha + \eta^\alpha_s) (g^\alpha_s(\C)-g^\alpha_s(b^k)) \right]
\end{equation}

\eqref{eq:recursion-pre} suggests that an inductive procedure for deriving the partition function can be formulated. In order to continue the induction procedure, we define the weight function of a contour as 
\begin{equation}\label{eq:weight-function}
    w_\eta^{k_0}(\C) = \exp\left[-\frac{1}{T}D^{k_0}_\eta(\C) \right] \prod_k \frac{\tilde{\PF}^k_{\eta, \Int_k \C}}{\PF^{k_0}_{\eta, \Int_k \C} 
    }  .  
\end{equation}
Applying the calculation of \eqref{eq:recursion-pre} inductively to $\Reg = \Int_k \C$ yields the following formula: 

\begin{proposition}[Polymer formula]\label{prop:polymer-formula}

For the modified partition function $\Xi^{k_0}_{\eta,\Reg} =  \exp[(\Ge|\Reg|+ S^{k_0}_\Reg)/T] \PF^{k_0}_{\eta,\Reg}$, an expansion in terms of the contour polymer can be derived  
\begin{equation}\label{eq:PF-polymer-formula}
    \Xi^k_{\eta, \Reg} = \sum_{\TB^k} \prod_{\C \in \TB^k}  w_\eta^k(\C)  \prod_{\C_i, \C_j, i \neq j} \delta(\C_i, \C_j),
\end{equation}
where the contour polymer is defined by $\TB^k$ ranging in all the possible collections of contours in $\CoC^k$ and the compatibility condition $\delta({\C_i, \C_j})$ as 
\begin{equation}\label{eq:polymer-metric}
    \delta(\C, \C') = \left\{ \begin{array}{cc}
        1 & d(\C, \C')>1, \text{ or } d(\C, \C') =1 \text{ and } \C \ge \C'  \text{ or } \C \ge \C' \\
        0 & \text{otherwise}
    \end{array} \right..
\end{equation}

Moreover, we can derive a polymer formula of the occurrence of a fixed external contour $\C_0$ in the Gibbs measure $\mu^k_{\eta, \Reg}$ as 
\begin{equation}\label{eq:contour-prob-formula}
\mu^k_{\eta, \Reg}(\C_0 \in \TB_\Ext(x)) = \frac{\sum_{\TB^k, \C_0 \in \TB^k_\Ext } w^k_\eta(\C_0)  \prod_{\C \in \TB^k\setminus\{\C_0\} } w^k_\eta(\C)  \prod_{\C, \C' \in \TB^k, \C \neq \C'} \delta(\C, \C')}{ \sum_{\TB^k} \prod_{\C \in \TB^k}  w^k_\eta(\C)  \prod_{\C, \C' \in \TB^k, \C \neq \C'} \delta(\C, \C')} \\
,
\end{equation}
where the weight function is defined in (\ref{eq:weight-function}). 
\end{proposition}

\section{The Ding-Zhuang argument}\label{sec:Ding-zhuang}
\subsection{General analysis of the weight function}\label{sec:Ding-Zhuang-general}

To demonstrate stability under weak random fields, the core strategy of the Ding-Zhuang argument embodies the introduction of symmetric transformations $\tau^{k_0 \rightarrow k_1}_\Reg$ and $\bar{\tau}^{k_0 \rightarrow k_1}_\Reg$ acting on spin configurations and the random field configurations. Within the Pirogov-Sinai contours framework, our generalized Ding-Zhuang argument decomposes into steps targeting these three key objects:
\begin{enumerate}
    \item $D^k_\eta(\C)$, the modified excitation energy of the contour under additional external fields.
    \item $\tilde{\PF}^{k_1}_{\eta, \Reg} / \PF^{k_0}_{\tau^{k_0 \rightarrow k_1}_{\Reg}(\eta), \Reg}$, representing quasi-invariance under symmetric operations on the Hamiltonian.
    \item $\PF^{k_0}_{\tau^{k_0 \rightarrow k_1}_{\Reg^\prime}(\eta), \Reg}/\PF^{k_0}_{\eta,\Reg}$, the partition function ratio for local random field transformations under fixed boundary conditions.
\end{enumerate}

In the following, we denote the symmetric transformation by $\tau_\Reg$ for simplicity when the superscripts are contextually clear. The main estimates for the long-range order phase focus on achieving rapid decay in the probability of contour $\C$ occurrence. Specifically, the probability must decay at least exponentially to dominate the entropy effects arising from counting the contours, where contour cardinality grows at most exponentially with contour size. Here the exponential rates must be sufficiently large under both low temperature and weak random field conditions. To obtain satisfactory exponential decay, we design the procedure to address the three key objects:

\textit{Step 1: Estimates of $D^k_\eta(\C)$ and $\FSIR$.} To bound the first term $D^k_\eta(\C)$, we introduce the classical Peierls condition from the Peierls argument, which characterizes the excitation energy induced by a contour in the unperturbed Hamiltonian. 

\begin{definition}[Peierls Condition]\label{def:Peierls-Condition}
The unperturbed statistical system with Hamiltonian $H^k_{0,\Reg}$ satisfies the Peierls condition if: 
(i) the number of periodic local ground states is finite, and (ii) there exists $\rho>0$ such that for any contour $\C$ and finite region $\Reg \supset \sC$, $H_{0}(\C) - |\sC|\Ge \ge \rho |\sC|$.
\end{definition}

In the presence of random fields, however, the excitation energy typically fails to satisfy the Peierls condition almost surely. Specifically, random field fluctuations can alter the ground state from its unperturbed configuration with probability one in many cases. Nevertheless, for contours $\C$ with $0 \in \sC \cup \Int \C$, we establish that the Peierls condition holds with modified constant $\rho' < \rho$ with high probability. This probabilistic guarantee relies fundamentally on subgaussian concentration of the random field-induced energy changes. To formalize this stability, we define:

\begin{definition}[Fluctuation-stabilized contour]
A contour $\C$ is fluctuation-stabilized for the random field configuration $\eta$ if $|D^k_\eta(\C) - D^k_0(\C)|\le \rho|\sC|/4$. Denote by $\FSC$ the subset of the random field configuration $\eta$ for which every contour $\C$ satisfying $0\in \sC \cup \Int \C$ is fluctuation-stabilized.
\end{definition}

\textit{Step 2: Estimates of $\tilde{\PF}^{k_1}_{\eta, \Reg} / \PF^{k_0}_{\tau_{\Reg} \eta, \Reg}$ and $ \QISC$}. The second term is bounded by comparing the summands of $ \tilde{\PF}^{k_1}_{\eta, \Reg}$ to $\PF^{k_0}_{\tau_{\Reg} \eta, \Reg}$ using the local symmetry operations. Although the differing boundary requirements in their definitions cause inherent non-equivalence, this distinction enables analysis of  pseudo-symmetry. Consequently, a proper symmetry operation yield the bound. In addition, to utilize the tensorization property, the transformation $\eta \mapsto \tau_\Reg\eta$ is induced by operations on the quenched parameters $\omega$. In other words, $\tau_\Reg\eta (\omega):= \eta (\tau_\Reg \omega)$. Thus, the symmetry condition is formalized as follows.

\begin{definition}[Local symmetry]\label{def:local-symmetry}
For a statistical physical model with random fields, there exists a local symmetry for any pair of local Hamiltonians $H^{k_1}_{\eta, \Reg}$ and $H^{k_2}_{\eta, \Reg}$ if there exists a continuous configuration transformation $\bar{\tau}_{\Reg}: \SV^{\Reg\cup \ExB \Reg} \rightarrow \SV^\Reg$ and a continuous quenched parameter configuration transformation $\tau_\Reg: \QP^{\Nb \times \Reg \cup \ExB_2 \Reg} \rightarrow \QP^{\Nb \times \Reg \cup \ExB \Reg}$ for any finite region $\Reg$ such that:
\begin{enumerate}
    \item (Locality) For any spin configuration $x$ (quenched parameter configuration $\omega $), $(\bar{\tau}_\Reg x)_s$ ($(\tau_\Reg \omega)_s$) with $s \in \Reg$ ($s \in \Reg \cup \ExB \Reg$) depends only on $x_t$ ($\omega_t$) for any $|t-s|\le 1$.
    \item (Injectivity) $\bar{\tau}_\Reg$ is injective over the subspace of the spin configuration $x$ with $x_s = b^{k_1}$ for any $s \in \ExB \Reg \cup \InB \Reg$.
    \item (Energy quasi-invariance) For any spin configuration $x$ with $x_s = b^{k_1}$ for any $s \in \InB_2 \Reg \cup \ExB \Reg$ and any quenched parameter configuration $\omega$, $|H^{k_1}_{ \eta(\omega), \Reg}(x) - H^{k_2}_{\eta (\tau_{\Reg}\omega), \Reg}(\bar{\tau}_\Reg x)| \le \sum_\alpha \sum_{s \in \InB \Reg} |\eta^\alpha_s |+|(\tau_\Reg\eta)^\alpha_s|$. 
    \item (Regularity) $\tau_\Reg$ is a Lipschitz function with a Lipschitz constant $C_\tau$ with respect to each quenched parameter $\omega_s$ on which it depends. %with constant $C_\tau$ over any single quenched parameter $\omega_s$ in dependence on.
    \item (Measure quasi-invariance) For any quenched parameters $\omega$, there exists a subspaces $P_\Reg$, consisting of pairs of form $(\alpha, s) \in N_\alpha \times (\Reg \cup \ExB_2 \Reg)$ such that $N_{\alpha}\times  (\Reg \cup \ExB_2 \Reg) \setminus P_\Reg \subseteq N_\alpha \times (\ExB_2 \Reg\setminus \ExB \Reg)$ and the induced measurable transformation $\tau_\Reg: (  \QP^{P_\Reg}, \bigotimes_{P_\Reg} (\MoQP^\alpha)_s) \rightarrow ( \QP^{N_\alpha \times (\Reg \cup \ExB \Reg ) }, \bigotimes_{N_\alpha \times (\Reg \cup \ExB \Reg )} (\MoQP^\alpha)_s )$ is well-defined and measure invariant.
\end{enumerate}
\end{definition}

From the energy quasi-invariance condition and \eqref{eq:quenched-condition}, for $\eta \equiv 0$ ($\omega \equiv 0$), we have $H^{k_1}_{0,\Reg}(x) = H^{k_2}_{0,\Reg}(\bar{\tau}_\Reg x)$. The ground state definition implies $H^{k_1}_{0,\Reg}(b^{k_1}) = H^{k_2}_{0,\Reg}(b^{k_2})$, so $\bar{\tau}_\Reg(b^{k_1}) = b^{k_2}$ in $\Reg$. By the locality condition, $\bar{\tau}_\Reg(x)_s = b^{k_2}$ whenever $x_t = b^{k_1}$ for some $|t-s| \leq 1$. This provides a local transformation between ground states. Then By the injectivity condition, $(\bar{\tau}_\Reg x)_s = b^{k_2}$ for $s \in \InB_2 \Reg$ if $x_s = b^{k_1}$ for $s \in \InB_3 \Reg$. Returning to the estimates of the partition function ratio, the local symmetry operation defines an injective map from summands in $\tilde{\PF}^{k_1}_{\eta,\Reg}$ to those in $\PF^{k_0}_{\tau_{\Reg}\eta,\Reg}$, bounding the ratio via the energy quasi-invariance condition as 
\begin{equation}
    \frac{\tilde{\PF}^{k_1}_{\eta, \Reg}}{  \PF^{k_0}_{\tau_{\Reg} \eta, \Reg}} \leq \exp \left[ \frac{1}{T} \sum_\alpha \sum_{s \in \InB \Reg} |\eta^\alpha_s |+|(\tau_\Reg\eta)^\alpha_s| \right].
\end{equation}
 
In application, we consider local symmetries on internal regions of contours. Due to the lack of strict invertibility of the local symmetry, the quasi-invariance of energy is inevitable in many cases. As before, we cannot make the energy difference less than $\rho |\sC|$ almost surely. Nevertheless, we can still apply the subgaussian estimations analogously to confirm the quasi-invariance-stability with high probability, formalized as:

\begin{definition}[Quasi-invariance-stabilized contour]
A contour $\C$ is quasi-invariance-stabilized for the random field configuration $\eta$ if $\sum_\alpha \sum_{s \in \InB \Int \C} |\eta^\alpha_s |+|(\tau_\Reg\eta)^\alpha_s| \le \rho|\sC|/4$. Denote by $\QISC$ the subset of the random field configuration $\eta$ for which every contour $\C$ satisfying $0\in \sC \cup \Int \C$ is quasi-invariance-stabilized.
\end{definition}

\textit{Step 3: Estimates of $\PF^{k_0}_{\tau_{\Reg^\prime}\eta, \Reg}/\PF^{k_0}_{\eta,\Reg}$ and $\FSIR$.} The analysis of the third term $\PF^{k_0}_{\tau_{\Reg^\prime}\eta, \Reg}/\PF^{k_0}_{\eta,\Reg}$ constitutes the most technically demanding component of the Ding-Zhuang argument. We rewrite the problem in the form of the free energy difference
\begin{equation}\label{eq:diff-of-free-energy-fml}
\Delta_{\Reg^\prime} F^{k}_\Reg := T\left(\ln \PF^{k}_{\eta ( \tau_{\Reg^\prime} \omega),\Reg}-  \ln \PF^{k}_{\eta(\omega),\Reg } \right).
\end{equation}

In the estimates of probability of the contour occurrence, we focus the case $\Reg^\prime = \Int \C$, with the corresponding local transformation $\tau_{\Int \C}$ the composition of the transformations $\tau^{k \rightarrow k_0}_{\Int_k \C}$ for the pair of Hamiltonian $H^{k}_{\eta, \Int_k \C}$ and $H^{k_0}_{\eta, \Int_k \C}$. A similar yet more tricky stabilization of the internal region is expected. To formalize this notion, we introduce the following definition analogous to that of $\FSC$ as:

\begin{definition}[Fluctuation-stabilized internal region]
The internal region $\Int \C$ of a contour $\C$ is fluctuation-stabilized for the random field configuration $\eta$ if $\Delta_{\Int \C} F^{k}_\Reg \le  \rho|\sC|/4$ with respect to the local quenched parameter transformation $\tau_{\Int \C}$.
Denote by $\FSIR$ the subset of the random field configuration $\eta$ such that for all contour $\C$ with $0\in \sC \cup \Int \C$, the internal region is fluctuation-stabilized. 
\end{definition}

The major distinction in estimating $\Delta_{\Int \C} \tilde{F}^{k}_\Reg$ lies in its dependence on quenched parameters whose number scales superlinearly with $|\sC|$. This fundamentally contrasts with fluctuation-stabilized contours, where only $O(|\sC|)$ parameters are involved. To establish the high probability of $\FSIR$, we employ the classical chaining method for supremum estimation. This technique requires two key prerequisites:
(i) Subgaussian control of variations across domains $\Reg^\prime$ (Section \ref{sec:subgaussian})
(ii) Coarse-grained processes measuring geometric properties of internal regions of contours (Section \ref{sec:coarse-graining}).
\subsection{The subgaussian estimate and its applications}\label{sec:subgaussian}

In the subsequent proof, we frequently need to bound important quantities induced by the identically and independently distributed quenched parameters $\omega$. We aim to show that their values converge with high probability to a region around their mean. Specifically, this convergence occurs in a sub-gaussian manner. In this article, we only pick up the results on two typical cases of distributions over the single quenched parameter, the bounded distribution and the normal distribution. We introduce the following two theorems:

\begin{theorem}[McDiarmid's Inequality \textup{\cite[Theorem 3.11]{APC550}}]
\label{thm:bounded-concentration-origin}
Let $\{X_i\}_{i=1}^n$ be independent random variables. Let $f: \prod_{i=1}^n X_i \to \mathbb{R}$ be a function such that  
\begin{equation}
   D_f:= \max_i (\max_{\{X_j\}_{j\neq i}} (\max_{X_i} f(\{X_s\}) - \min_{X_i}f(\{X_s\}))) < \infty. 
\end{equation}
Then, $f(\{ X_i \})$ is subgaussian with variance $ND_f^2/4$. In particular, for any $t > 0$:
\begin{equation}
\Pb \left[\left| f(\{ X_i \} ) - \Ex f(\{ X_i \} )\right| \ge \lambda \right]\leq 2\exp\left( -\frac{2\lambda^2}{n D_f^2} \right).
\end{equation}
\end{theorem}

\begin{theorem}[Gaussian concentration, \textup{\cite[Theorem 3.25]{APC550}}]
\label{thm:gaussian-concentration}
Let $\{X_i\}_{i=1}^n$ be independent and identical Gaussian random variables with $\Ex(X_i)=0$ and $\Ex(X_i^2) = \epsilon^2$ for all $i$. Then, for any function $f: \mathbb{R}^n \to \mathbb{R}$, and for all $t \geq 0$,
\begin{equation}
\Pb \left[\left| f(\{ X_i \} ) - \Ex f(\{ X_i \} )\right| \ge \lambda \right] \leq 2\exp\left( -\frac{\lambda^2}{2nD^2_f\epsilon^2} \right),    
\end{equation}
where $D_f^2 = \frac{1}{n}\left\| \|\nabla f\|^2 \right\|_{\infty}$.
\end{theorem}
For the bounded case, we further obtain a subgaussian coefficient that yields stronger probability suppression when the variation is small and the function to be estimated is Lipschitz continuous, which is crucial in the following analysis.

\begin{proposition}
\label{prop:rev-bounded-concentration}
Let $\{X_i\}_{i=1}^n$ be bounded independent random variables with $|X_i|\le 1$, $\Ex X^2_i \le \epsilon^2$ for all $i$. Let $f: \prod_{i=1}^n X_i \to \mathbb{R}$ be a Lipschitz function with constant $D_f$ for any variable $X_i$. Then, there exists $\nu:= \nu(\epsilon)>0$ such that $f(\{ X_i \})$ is subgaussian with variance $\nu n D_f^2/2$. In particular, for any $t > 0$:
\begin{equation}\label{eq:subgaussian-rev-bounded}
  \Pb \left[\left| f(\{ X_i \} ) - \Ex f(\{ X_i \} )\right| \ge \lambda \right] \le 2\exp\left( -\frac{ \lambda^2 }{2\nu n D_f^2 }\right), 
\end{equation}
where $\nu : = \nu(\epsilon)$ such that $\lim_{\epsilon \rightarrow 0} \nu  = 0$. 
\end{proposition}

Note that the same inequality \eqref{eq:subgaussian-rev-bounded} follows for Gaussian random variables with $\nu = \epsilon^2$ To verify this, we derive the following technical lemma.
\begin{lemma}\label{lem:single-log-generation-bound}
Given $\epsilon>0$. There exists $\nu:= \nu(\epsilon)>0$ such that for any random variable $X$ with $|X|\le 1 $, $\Ex X = 0$ and $\Ex |X|^2\le \epsilon^2$, 
\begin{equation}\label{eq:single-log-generation-bound}
    \Ex (e^{\lambda X}) \le e^{\frac{\nu \lambda^2 }{2}} 
\end{equation}
and $\lim_{\epsilon \rightarrow 0} \nu  = 0$.
\end{lemma}
\begin{proof}
Since $f(x) =  e^{\lambda x}$ is a convex function, the maximal value of $\mu(e^{\lambda X})$ for $\mu$ a probability measure over a convex subspace of the probability measure space over $[-1,1]$ can only be reached in its extremal points. In addition, for the probability measure with fixed one and two moments, the support of any extremal measure contains only $3$ points, by results over Chebyshev system \cite[Chap IV, Sec. 6]{Tchebycheff-systems}.  As a result, the problem is reduced to solve a finite dimensional optimization problem shown in lemma \ref{lem:opt-3point-generation} and
\begin{equation}
\Ex (e^{\lambda X}) \le \frac{e^{-\lambda \epsilon^2} + \epsilon^2 e^{\lambda}}{1 + \epsilon^2}
\end{equation}For $\lambda \ge 0$, we are trying to find a function $\nu:= \nu (\epsilon)$ so that \eqref{eq:single-log-generation-bound} should be satisfied with the required properties. An equivalent problem is to calculate the upper bound of 
\begin{equation}\label{eq:opt-to-quartic-upperbound}
\begin{aligned}
& \max_{\lambda\ge 0} \left[ \frac{2}{\lambda^2}\ln\left( \frac{e^{-\lambda \epsilon^2} + \epsilon^2 e^{\lambda}}{1 + \epsilon^2} \right)  \right] \\
\le & \max_{\lambda\ge 0} \left[ \frac{2}{\lambda^2}\ln\left( 1+ \lambda \epsilon^2 \frac{e^{\lambda} - e^{-\lambda \epsilon^2}}{1 + \epsilon^2} \right)  \right] \\
\le &  \max_{\lambda\ge 0} \left[ \frac{2}{\lambda^2 } \ln \left(1 + e^\lambda\lambda^2 \epsilon^2\right) \right].
\end{aligned}
\end{equation}
We denote $w : = e^\lambda\lambda^2 \epsilon^2$ and $h(\lambda) =  \frac{2}{\lambda^2 } \ln \left(1 + e^\lambda\lambda^2 \epsilon^2\right) $. By calculating the derivative, it is obvious to find that $h(\lambda)$ reaches its maximal value for some $\lambda_0>0$ for only one point satisfying the equation \begin{equation}\label{eq:lambda-maximizer}
(2+\lambda)w = 2(1+w)\ln (1+w). 
\end{equation}
Since $\lambda = \frac{2(1+w)\ln(1+w)}{w}-2 > 2 w$ for $\lambda = \lambda_0$, we have $\lambda_0 \ge O(|\ln \epsilon|)$ for sufficiently small $\epsilon$. Then we have $\lambda_0 = 2 \ln w +o(1)$ and $\lambda_0 = -4 \ln  \epsilon + o(|\ln \epsilon|)$. Then we have $h(\lambda) \le 1/|4\ln\epsilon| + o(1/|\ln \epsilon|)$ and the conclusion follows for $\nu = 1/|2\ln\epsilon|$ for sufficiently small $\epsilon$.  
\end{proof}

\begin{proof}[Proof of Proposition \ref{prop:rev-bounded-concentration}]
The subgaussian inequality follows by Azuma's lemma as Theorem \ref{thm:bounded-concentration-origin}. Since $\Ex(f|\{X_i\}_{i=1}^{k+1}) - \Ex(f|\{X_i\}_{i=1}^k)$ has absolute value bounded by $D_f$ due to Lipschitz continuity and variance upper bounded by $D_f^2\epsilon^2$ due to Lemma \ref{lem:variation-boundedness}, we can derive the conditional subgaussian property with a uniform coefficient $D_f^2 \nu$ of Azuma's lemma by rescalling and Lemma \ref{lem:single-log-generation-bound}. Then the conclusion follows. 
\end{proof}

We first apply the two theorems to the quantity $D^k_\eta(\C) - D^k_0(\C)$ to obtain the lower bound of the probability of fluctuation-stabilized contours $\FSC$. 

\begin{proposition}\label{prop:prob-of-FSC}
Let $\delta > 0$ be given. Then there exists $\epsilon_1: = \epsilon_1(\delta) > 0$ such that if the maximal standard variation $\epsilon < \epsilon_1$, then for $d \ge 2$
\begin{equation}
\mathbb{P}(\eta \in \FSC) \geq 1 -  \delta.    
\end{equation}
\end{proposition}
\begin{proof}
Decompose the event that $\eta$ fails to fluctuation-stabilize contours containing the origin:  
\begin{equation}
\{\eta \notin \FSC \} \subseteq \bigcup_{\C : 0 \in \sC \cup \Int \C}  \left\{ |D^k_\eta(\C) - D^k_0(\C)| > \frac{\rho}{4} |\C| \right\}.  
\end{equation}
Group contours by boundary length. Let $\CoC^k_0(n)$ is the collection of all possible contours $\C \in \CoC^k$ with $0\in \sC \cup \Int \C$ and $|\sC|=n$ and similarly define $\CoC_0(n)$. Then for fixed $\C \in \CoC^k_0(n)$, the energy difference deviation is:  
\begin{equation}
\Delta D^k_\eta(\C) : = D^k_\eta(\C ) - D^k_0(\C) = \sum_\alpha \sum_{s \in  \sC}  \eta_s^\alpha(\omega)(g^\alpha(\C)-g^\alpha(b^k)).  
\end{equation} Using the locality of $\eta_s^\alpha$ depending only on $\omega_t^\beta$ for $|s-t| \leq 1$, we already know that $\Delta D^k_\eta(\C)$ depends only on $\omega^\beta_s$ with $s \in \sC \cup \ExB \sC$. Since $E[\omega_0^\beta] = 0$, $\Delta H_\eta(\C)$ is Lipschitz-$2$ with $\omega^\beta_t$ by the boundedness and Lipschitz condition of $g^\alpha_s$. By the independence of $\omega $, the subgaussian concentration inequality follows from Proposition \ref{prop:rev-bounded-concentration} or Theorem \ref{thm:gaussian-concentration}:  
\begin{equation}
\Pb \left( |\Delta D^k_\eta(\C) - \Ex(\Delta D^k_\eta(\C)) | > \frac{\rho}{4} |\sC| \right) \leq \exp\left( - \frac{\rho^2|\sC|^2}{ 128 \nu|\sC \cup \ExB \sC|} \right),  
\end{equation} where $\nu: =\nu(\epsilon)$ as in the results of concentration. The expectation of $\Delta H_\eta$ is vanishing due to $\Ex \eta_s = 0$. Since $|\sC \cup \ExB \sC| \le 3^d |\sC |$, summing up the probability bound over the contours, we get  
\begin{equation}
\mathbb{P}(\eta \notin \FSC) \leq \sum_{n=1}^\infty | \CoC_0(n)| \exp\left( - \frac{\rho^2 n }{128 \nu3^d } \right) \leq  \sum_{n=1}^\infty  \exp \left[ - n \left(  
 \frac{\rho^2  }{128 \nu3^d   } -d\ln{9} -1- \ln\Ns \right) \right].      
\end{equation} The last inequality follows from Lemma \ref{lem:combinatorics-contour}. As a result, the right hand side is bounded by $\delta$ when $\epsilon$ is sufficiently small and then $\mathbb{P}(\eta \in \FSC) > 1 - \delta$.  
\end{proof}

Similar to the proof of Proposition \ref{prop:prob-of-FSC}, we can also derive the probability bound of $\QISC$ by estimating $\sum_\alpha \sum_{s \in \InB \Int \C} |\eta^\alpha_s |+|(\tau_\Reg\eta)^\alpha_s|$.
\begin{corollary}\label{cor:prob-of-QISC}
Let $\delta > 0$ be given. Then there exists $\epsilon_2: = \epsilon_2(\delta) > 0$ such that if the maximal standard variation $\epsilon < \epsilon_2$, then for $d \ge 2$
\begin{equation}
\mathbb{P}(\eta \in \QISC) \geq 1 - \delta.    
\end{equation}\end{corollary}
\begin{proof} 
The proof is reproduced almost verbatim by that of Proposition \ref{prop:prob-of-FSC}. One of the differences is a larger dependence of the estimated quantities over $\omega^\beta_s$ for $s \in \ExB_2 \Int \C \cup \InB_3 \Int \C$, the area of which is of order $O(\sC)$. The other difference is that the expectation $ \sum_\alpha \sum_{s \in \InB \Int \C} \Ex|\eta^\alpha_s |+\Ex|(\tau_\Reg\eta)^\alpha_s|$ is nontrivial. However, each term is at most proportional to its standard deviation, which is of order $O(\epsilon)$ due to Lemma \ref{lem:variation-boundedness} with constants given in the regularity condition. As a result, we can use the subgaussian concentration inequality with a smaller $\lambda$ and the conclusion follows.  
\end{proof}
For the probability bound of the fluctuation-stabilized internal regions $\FSIR$, The dependence of $\tau_\Reg$ over the boundary condition of the targeted Hamiltonian is omitted since the following argument is satisfied in all such cases. The main quantity focused on is the difference of the free energy $
\Delta_{\Reg^\prime} \tilde{F}^{k}_\Reg$ given in \eqref{eq:diff-of-free-energy-fml}. At first sight, the fluctuation depends on a larger number of quenched parameters than the case of fluctuation-stabilized contours and quasi-invariance-stabilized contours for every single local symmetric transformation. Then the same strategy cannot be applied here due to the relatively fat-tail distribution. For this case, to measure the maximal deviations of a collection of quantities, the chaining method plays its role and the following lemma is necessary to measure the relation of the difference between different quantities $\Delta \tilde{F}^k_\Reg$ and the difference between different defining domains $\Reg^\prime$.   
\begin{lemma}\label{lem:subgaussian-area-variation}
 For any finite region $\Reg_1$, $\Reg_2 \subseteq \Reg$, the quantity $\Delta_{\Reg_1} \tilde{F}^{k}_\Reg- \Delta_{\Reg_2} \tilde{F}^{k}_\Reg$ obeys the following subgaussian concentration inequality for some constants:
\begin{equation}\label{eq:subgaussian-area-variation}
\Pb(|\Delta_{\Reg_1} \tilde{F}^{k}_\Reg- \Delta_{\Reg_2} \tilde{F}^{k}_\Reg| \ge  \lambda )\le \exp\left( -\frac{\lambda^2}{\nu (|\Reg_1 \Delta \Reg_2|+ |\ExB_2 \Reg_1 \cup \ExB_2 \Reg_2|)} \right),  
\end{equation}
where $\Reg_1 \Delta \Reg_2$ denotes the symmetric difference between $\Reg_1$ and $\Reg_2$. 
\end{lemma}\begin{proof}
Since $|\partial_{\omega_s }  \Delta_{\Reg_i} \tilde{F}^{k}_{\Reg} |\le \max_x \{|\partial_{\omega_s}H_{\eta, \Reg}(x)|+|\partial_{\omega_s}H_{\tau_{\Reg_i} \eta, \Reg}(x)|\} \le 2$ by the direct calculation, we can obtain the concentration inequality with a proper conditional expectation of the random variable. Set 
\begin{equation}\label{eq:area-fixed}
    \Reg^\prime = (\Reg \setminus (\Reg_1 \cup \ExB_2 \Reg_1 ) \cap \Reg \setminus (\Reg_2 \cup \ExB_2 \Reg_2)) \cup (\Reg_1  \cap \Reg_2).
\end{equation} Then due to the locality condition of the local symmetry transformation $\tau$, $(\tau_{\Reg_1 } \eta)_s = (\tau_{\Reg_2} \eta)_s$ for all $s \in \Reg^\prime$.  Since $\Ex(\Delta_{\Reg_1} \tilde{F}^{k}_\Reg- \Delta_{\Reg_2} \tilde{F}^{k}_\Reg| \QP^{N_\alpha \times \Reg^\prime})= 0$ also for any region $\Reg$ by the measure quasi-invariance condition and $N_\alpha \times \Reg' \subseteq P^{\Reg_1} \cap P^{\Reg_2}$, we have the following concentration inequality by applying Proposition \ref{prop:rev-bounded-concentration} or Theorem \ref{thm:gaussian-concentration}, 
\begin{equation}
    \Pb\left(|\Delta_{\Reg_1} \tilde{F}^{k}_\Reg- \Delta_{\Reg_2} \tilde{F}^{k}_\Reg | \ge  \lambda \left|  \QP^{N_\alpha \times \Reg^\prime}\right. \right) \le \exp\left(-\frac{\lambda^2}{\nu |\Reg \setminus \Reg^\prime|}\right).
\end{equation}
Then the conclusion follows by taking the unconditional expectation.
\end{proof}

\subsection{The chaining method}\label{sec:chaining}
To initiate the multiscale analysis of $\FSIR$, we employ the chaining method, which has become a popular tool in maximum estimation. To apply these methods, recall that $\CoC^k_0(n)$ is the collection of all contours $\C \in \CoC^k$ satisfying $0\in \sC \cup \Int \C$ and $|\sC|=n$. It is easy to find that for any connected subregions $\Reg$, we can define the internal region $\Int \Reg $ as for contours in Section \ref{sec:contour-concept}. Therefore, the supports $\sC$ of contours in $\CoC^k_0(n)$ are in the finite collection of all connected subregions $\Reg$ such that $0 \in \Reg \cup \Int \Reg$ and $|\Reg| = n$. Denote this collected by $\bar{\CoC}_0(n)$. For convenience, we also denote by the elements and their internal regions of $\bar{\CoC}_0(n)$ as $\sC$ and $\Int \C$. We then endow $\bar{\CoC}_0(n)$ with a metric space structure $(\bar{\CoC}_0(n),\bar{d})$, where the distance $\bar{d}(\sC_1, \sC_2)$ is defined by 
\begin{equation}\label{eq:distance_CoC}
\bar{d}(\sC_1, \sC_2)= \left\{
\begin{array}{cc}
    \nu \left[\sum_{k=1}^\Ng \left(|\Int_k \C_1 \Delta \Int_k \C_2|+ |\ExB_2  \Int_k \C_1 \cup \ExB_2 \Int_k \C_2 |\right) \right]^{1/2} & \sC_1 \neq \sC_2  \\
    0  & \sC_1 = \sC_2 
\end{array} \right..   
\end{equation}

The following definition of subgaussian process takes the central part in formalizing the chaining method:
\begin{definition}[subgaussian process]
\label{def:subgaussian-process}
A family of random variables $(X_u)_{u\in T}$, labelled by elements of a metric space $(T, d)$, is called subgaussian if there exists a constant $\sigma>0$ such that for all $\lambda>0$ and $u, v \in T$,
\begin{equation}\label{eq:def-subgaussian}
       \Pb \left( |X_u - X_v| \geq \lambda \right) \leq  2\exp{\left( \frac{-\lambda^2}{\sigma^2 \mathrm{d}(u,v)^2} \right)}.
\end{equation}
\end{definition}

It is easily seen from Lemma \ref{lem:subgaussian-area-variation} that $\Delta_{\Int \C} \tilde{F}^{k}_\Reg$ are sub-gaussian random variables with respect to $\sC \in \bar{\CoC}_0(n)$ with the metric $\bar{d}$. To apply the chaining method, we introduce a cover of the whole metric labeled by a subset called $r$-net: 
\begin{definition}[$r$-net and covering number]A set $N$ is called an $r$-net for $(T,d)$ if for every $t\in T$, there exists $\pi(t)\in N$ such that $d(t,\pi(t))\leq r$. The smallest cardinality of an $r$-net for $(T,d)$ is called the covering number, defined as 
\begin{equation}
    N(T,d,r):=\inf\{|N|:N\mbox{ is an } r \mbox{-net for }(T,d)\}.
\end{equation}
\end{definition}
A famous estimation then follows from the following theorem. It gives another subgaussian estimation useful in the maximum control. 
\begin{theorem}[Chaining tail inequality, \textup{ \cite[Theorem 5.29]{APC550}}]\label{thm:chaining-Tail-inequality}
Let $\{X_{t}\}_{t\in T}$ be a separable sub-gaussian process on the metric space $(T,d)$. 
Then for all $t_{0}\in T$ and $x\geq 0$
\begin{equation}\label{eq:chaining-tail-ineq}
\Pb\biggl{(}\sup_{t\in T}\{X_{t}-X_{t_{0}}\}\geq C\int_{0}^{\infty} \sqrt{\log N(T,d, r)}\, dr+x\biggr{)}\leq C\exp \left(-\frac{x^{2}}{C\, \mathrm{diam}(T)^{2}} \right),    
\end{equation}
where $C<\infty$ is a universal constant and $\mathrm{diam}(T)$ represents the diameter of $T$ over $d$.
\end{theorem}
To complete the argument, note that we can easily find a proper $\sC_0$ such that $\Delta_{\Int \C} \tilde{F}^{k}_\Reg = 0$ if $\Int \C_0 =\emptyset$. Then the only term necessary to be focused on is the covering number $N(T,d,r)$. To generate a $r$-net, we will conduct the following coarse-graining process to partition the whole space in a gradually coarser way. The final result will provide a proper upper bound of the covering number.    

\subsection{The coarse-grained process}\label{sec:coarse-graining}
To construct the covering with balls, as employed in the classical chaining method, we apply the coarse-graining argument originally introduced in \cite{FFS84}. For each $\ell>0$ and each contour ${\C}$, we associate a region $B_\ell(\C)$ that approximates the interior $\Int \C$ in a scaled lattice, where the scale factor increases with $\ell$. The covering then is achieved by ensuring that two contours with identical representations reside within a ball of fixed radius, which depends on $\ell$. For any $s \in \Zd$ and $\ell \geq 0$, we construct a tilting of the lattice by the $\ell$-cubes 
\begin{equation}
    B_{\ell}(s) := \left(\prod_{i=1}^d{\left[2^{\ell} s  , \ 2^{\ell}(s+1) \right)}\right)\cap \mathbb{Z}^d,
\end{equation}
centered at $2^{\ell}s + 2^{\ell-1} - \frac{1}{2}$ with side length $2^{\ell} -1$. When there is no ambiguity, we denote $B_\ell(s)$ simply by $B_{\ell}$. An arbitrary collection of $\ell$-cubes will be denoted $\CoB_\ell$ and $B_{\CoB_\ell}:= \cup_{B_\ell \in  \CoB_\ell} B_\ell$ is the region covered by $\CoB_\ell$. Fix $n, \ell \in \mathbb{Z}^+$ and $\sC \in \bar{\CoC}_0(n)$. An $\ell$-cube $C_{\ell}$ is admissible with respect to the region $\Reg$ if more than a half of its points are inside $\Reg$. In other words, we denote the set of admissible cubes with respect to $\Reg$ is
\begin{equation}
    \CoB_\ell(\Reg) := \{B_{\ell} : |B_{\ell}\cap \Reg| \geq \frac{1}{2}|\Reg|\}.
\end{equation}
We set $B_\ell(\Reg) := B_{\CoB_{\ell}(\Reg)}$, the region covered by the admissible cubes as a coarse-grained replica of $\Reg$.

To calculate the effects caused by the coarse-graining process, we need the following results to evaluate the difference in the area of the region and the boundary before and after the coarse-graining process. These arguments were first introduced without detailed proofs.
% These arguments are first researched without detailed proofs. 
To obtain complete proofs, see \cite{maia-2307.14150}.  

\begin{lemma}\label{Lem:Coarse-grained-1}
    Given $\Reg \subset \mathbb{Z}^d$, $\ell\geq 0$ and $U= B_{\ell}\cup B_{\ell}^\prime$ with $B_{\ell}$ and $B_{\ell}^\prime$ being two $\ell$-cubes sharing a face, there exists a constant $b_0:=b_0(d)$, dependent only on the dimension $d$, such that, if 
\begin{align}\label{Eq. U.condition}
    \frac{1}{2}|B_{\ell}| \leq |B_{\ell}\cap \Reg| \qquad \text{and} \qquad |B_{\ell}^\prime\cap \Reg|< \frac{1}{2}|B_{\ell}^\prime|
\end{align}
then $2^{\ell(d-1)}\leq b_0|\ExB \Reg\cap U|$.
\end{lemma}
Similarly, we can also derive the following proposition as in \cite{maia-2307.14150} with some minor changes.
\begin{proposition}\label{prop:Coarse-grained-2} For any finite region $\Reg$ and the corresponding coarse-grained regions $B_0(\Reg),\dots,B_k(\Reg)$ defined above, there exists constants $b_1:=b_1(d),b_2:=b_2(d)$ such that 
\begin{equation}\label{eq:coarse-grained-boundary}
    |\ExB B_\ell(\Reg)| \leq b_1|\ExB \Reg |
\end{equation}
and 
\begin{equation}\label{eq:coarse-grained-internal-region}
    |B_\ell(\Reg)\Delta B_{\ell+1}(\Reg)| \leq b_2 2^{\ell} |\ExB \Reg|
\end{equation}
for every $\ell\in\{0,\dots,k\}$.
\end{proposition}

Then we can derive the following results considering the covering number:
\begin{corollary}\label{cor:covering-number}
For any $n \in \mathbb{Z}^+$, there exists $\ell_0(n)>0$ such that $B_\ell(\Reg) = \emptyset$ for any region $\Reg$ with $|\Reg|=n$. For any $n \in \mathbb{Z}^+$, any $\ell>0$ and any two regions $\Reg_1,\Reg_2$ such that $\max\{|\ExB\Reg_1|,|\ExB\Reg_2|\} \le n $, $B_\ell(\Reg_1)=B_{\ell}(\Reg_2)$, there exists a constant $b_3:=b_3(d)>0$ such that 
\begin{equation}\label{eq:region-degeneration}
        d(\Reg_1,\Reg_2)\leq \nu b_3 2^{\frac{\ell}{2}} n^{\frac{1}{2}}. 
\end{equation}
where $d(\Reg_1,\Reg_2):= \nu|\Reg_1 \Delta \Reg_2 |^{1/2}$ represents the metric structure over the subset of $\Zd$. Subsequently, given $B_\ell(\bar{\CoC}_0(n)): = \{ \{B_\ell(\Int_k \C )\}^{\Ng}_{k=1}| \sC \in \bar{\CoC}_0(n) \}$ , the covering number of \((\bar{\CoC}_0(n), \bar{d})\) can be bounded as
\begin{equation}\label{eq:covering-number}
N(\bar{\CoC}_0(n), \bar{d},  \nu b_3 2^{\frac{\ell}{2}} n^{\frac{1}{2}}) \le |B_\ell(\bar{\CoC}_0(n))|
\end{equation}
\end{corollary}

\begin{proof}
The first statement follows easily from \eqref{eq:coarse-grained-boundary} since if $B_\ell(\Reg) \neq \emptyset$ then $|\ExB B_\ell(\Reg)|\ge (3^d-1) 2^{\ell(d-1)} $ which should be greater than $(3^d-1)b_1n = (3^d-1)b_1 |\Reg|\ge b_1|\ExB \Reg| $ for $\ell > \ell_0:= \ln(b_1 n)/(d-1)\ln (2)$.

The second statement follows a simple application of the triangular inequality, since 
\begin{equation}
    d(\Reg_1,\Reg_2) \leq d(\Reg_1 ,B_\ell(\Reg_1)) + d(\Reg_2,B_\ell(\Reg_2)) + d(B_\ell(\Reg_1), B_\ell(\Reg_2))
\end{equation} and 
\begin{equation}
  \begin{aligned}
        d(\Reg_i,B_\ell( \Reg_i)) &\leq \sum_{i=1}^\ell d(B_i(\Reg_i),B_{i-1}(\Reg_i)) = \sum_{i=1}^\ell \nu\sqrt{B_i(\Reg_i)\Delta B_{i-1}(\Reg_i)} \\
        & \leq \sum_{i=1}^\ell \nu\sqrt{b_2} 2^{\frac{i}{2}} \sqrt{|\ExB\Reg_i|}  \leq \nu\sqrt{b_2}\sqrt{2}(\sqrt{2}+1)2^{\frac{\ell}{2}} \sqrt{|\ExB\Reg_i|} ,
\end{aligned}
\end{equation}where in the second to last equation \eqref{eq:coarse-grained-internal-region} is used. Since $B_\ell(\Reg_1)=B_\ell(\Reg_2)$, $d(\Reg_1, \Reg_2) \le 2(\sqrt{2}+1)\sqrt{2b_2}\nu2^{\frac{\ell}{2}} \sqrt{n}$. Then the second statement follows by taking ${b_3 := 2(\sqrt{2}+1)\sqrt{2b_2}}$. 

For the last statement, since 
\begin{equation}
\begin{aligned}
    \bar{d}(\sC_1, \sC_2) \le & \nu b_3 2^{\frac{\ell}{2}} \max \left\{ \left(\sum_k |\ExB\Int_k \C_1| \right)^{\frac{1}{2}}, \left(\sum_k|\ExB\Int_k \C_2| \right)^{\frac{1}{2}}\right\}+ \nu (2n5^{d} )^{\frac{1}{2}} \\
    \le & \nu (b_3 + \sqrt{2} 5^{\frac{d}{2}}) 2^{\frac{\ell}{2}} n^{\frac{1}{2}}
\end{aligned}
\end{equation} if $B_\ell (\Int_k \C_1 ) = B_\ell(\Int_k \C_2)$ for all $k$ by applying the same argument for the second statement. Then pick up an element $\sC \in \bar{\CoC}_0(n)$ corresponding to every distinct $k$-tuple $\left(B_\ell(\Int_k \C )\right)^{\Ng}_{k=1}$ and the covering can be constructed with the radius $ \nu b_3 2^{\frac{\ell}{2}} n^{\frac{1}{2}}$ by enlarging $b_3$ by $\sqrt{2} 5^{\frac{d}{2}}$. 
\end{proof}

In the next proposition we bound $|B_\ell(\bar{\CoC})|$.
\begin{proposition}\label{prop:entropy-control}
    There exists a constant $b_4:= b_4(d)$ such that, for any $n\in\mathbb{N}$, \begin{equation}\label{eq:entropy-control}
        |B_\ell(\bar{\CoC}_0(n))|\leq \exp{\left\{\frac{b_4\ell n}{2^{\ell(d-1)}}\right\}},
    \end{equation}
\end{proposition}
\begin{proof}
The idea of the proof is similar to that of the appendix in \cite{FFS84}. See also \cite{maia-thesis}. For each $B_\ell(\Int_k \C)$, we denote $E^k_\ell(\C) := \ExB B_\ell(\Int_k \C )$. We might write $E^k_\ell$ instead of $E^k_\ell(\C)$ for simplicity. Also, denote the connected components of $E^k_\ell$ by $E^{k,(i)}_\ell$ for $1\le i\le m_k$, where $m_k$ represents the number of these connected components. Also denote $I^k_\ell:= \InB B_\ell(\Int_k \C )$. By Lemma \ref{Lem:Coarse-grained-1}, there are at least $b_0^{-1}2^{\ell (d-1)}$ points of $\ExB\Int_k \C$ in $B_\ell\cup B_\ell^\prime$ for any pair $(B_\ell, B'_\ell)$ such that $B_\ell \subset I^k_\ell$ and $B'_\ell\subset E^k_\ell$. Hence $m_k \le \frac{b_0 |\ExB \Int_k \C |}{2^{\ell(d-1)}}$ and the number of connected components for $\{E^k_\ell\}^{\Ng}_{k=1}$ is bounded by $\frac{b_0 n}{2^{\ell(d-1)}}$. Also, $|E^k_\ell| \le b_1 |\ExB \Int_k \C| \le b_1 n$ by Proposition \ref{prop:Coarse-grained-2}. 

We will estimate $|B_\ell(\bar{\CoC}_0(n))|$ by the following steps:

\textit{Step 1: Bound the contour number estimation with fixed inner points and length.}
we take $M_n:=\left\lfloor\frac{b_0n }{2^{\ell(d-1)}}\right\rfloor$. It is obvious that for $M_n = 0$, $B_\ell(\bar{\CoC}_0(n)) = \{\emptyset\}$ and the conclusion follows trivially. Thus we suppose $M_n \ge 1$. Fixed a set of points $\{s_i\}_{i=1}^{M_n}$ with $s_i \in 2^\ell\mathbb{Z}^d$, and $\{l_i\}_{i=1}^{M_n} $with $l_i \in 2^{\ell (d-1)}\mathbb{Z^+}\cup \{0\}$, set $N_\C(\{s_i\}_{i=1}^{M_n},\{l_i\}_{i=1}^{M_n})$ as the number of coarse-grained connected sets $E^k_\ell$ with $s_i \in  E^{k,(i)}_\ell$ and $|E^{k,(i)}_\ell| = l_i$ for every $i$, where we demand $s_i \in \emptyset$ always satisfied for $l_i=0$. Then its number bounds the number of all the $E^k_\ell(\C)$ corresponding to some $\C$ with the requirements given by $\{s_i\}, \{l_i\}$. Since $l_i$ has the form $2^{\ell(d-1)}z_i$ (where $z_i \in \mathbb{N}$)  and the upper bound for the number of Peierls contours of size $z_i$ is $3^{2z_id}$, we have:
\begin{equation}
         N_\C(\{s_i\}_{i=1}^{M^k_n},\{l_i\}_{i=1}^{M_n})  = \exp\left( 2d\ln(9) \sum_{i=1}^{M_n} \frac{l_i}{2^{\ell(d-1)}} \right)
         \leq \exp\left( 2b_1d \ln(9) \frac{n}{2^{\ell(d-1)}} \right)
\end{equation}
by Lemma \ref{lem:combinatorics-contour} and $\sum_{i=1}^{M_n} l_i \le n b_1 $. 

\textit{Step 2: bound the number of  possible collections of $\{l_i\}$ and the label of $k$.}
Since $\sum_{i=1}^{M_n} l_i \leq b_1 n$, the number of possible collections of $\{l_i\}$ is bounded by $2^{\frac{b_1n }{2^{l(d-1)}}}$. Similarly, since $M_n \le \frac{b_0 n}{2^{\ell(d-1)}}$, the number of ways that label the $E^{k,(i)}_\ell$ is bounded by $\Ng^{\frac{b_0 n }{2^{l(d-1)}}} $. Then the product of the two numbers are bounded by $\exp\left(\frac{C_1 n}{2^{\ell (d-1)}}\right)$. 

\textit{Step 3: bound the number of possible collections of $\{s_i\}$.}
Set $d_i := |s_i - s_{i-1}|$ for $i\ge 2$ while $d_1 = |s_1|$. we also choose $\{t_i \}_{i=1}^{M_n}$ with $t_i \in \Int \C $ such that $|s_i - t_i| = d(s_i, \Int \C )$, and $t_0= 0$. It is easily obtained that $d(s_i,t_i) \leq 2^{\ell}$. 
Since all the cubes are not connected, $d(t_i,t_{i-1})>2^\ell$, hence
\begin{equation}
    d(s_i,t_i) \leq |t_i -t_{i-1}|.
\end{equation}
As all $t_i \in \InB \Int \C$ with $i\ge 1$, we can reorder the terms to minimize the sum of distances and find a proper path connecting $\bar{t}_i \in \sC$ such that $|t_i-\bar{t}_1|=1$ ($i\ge 1$). Since the path can be obtained by a DFS algorithm in $\sC$ as Lemma \ref{lem:combinatorics-contour}, each point can be referred to at most $2(3^d-1)$ times. Thus from $d(t_0, t_1)\le n$, we get  
\begin{equation}
    \sum_{i=1}^{M_n} d(t_i, t_{i-1}) \leq   2(3^d-1)n + M_n +n < 2(3^d+b_0)n.
\end{equation}
This yields 
\begin{equation}\label{eq:bound-of-di}
    \sum_{i=1}^{M_n} d_i \leq 2\sum_{i=1}^{M_n} d(s_i,t_i) + \sum_{i=1}^{M_n} d(t_i,t_{i-1}) \leq 3\sum_{i=1}^{M_n} d(t_i, t_{i-1}) \leq 6(3^d +b_0) n.
\end{equation}
Denote $C_0 := 6(3^d+b_0)$. For fixed $\{d_i\}^{M_n}_{i=1}$, the number of ways of choosing $\{s_i\}$ is bounded by $\prod_{i=1}^{M_n} 2d(d_i)^{d-1} \le \prod_{i=1}^{M_n} (2d_i)^{d}$. 
By the AM–GM inequality, the maximum of this quantity is reached when all the distances are the same. Assuming $d_i\equiv d^*$, satisfying the following bound\begin{equation}
     d^* \leq \frac{C_0 n}{M_n} = \frac{2C_0 2^{\ell(d-1)}}{b_0},
 \end{equation}
the number is bounded by \begin{equation}
    \begin{aligned}
\prod _{i=1}^{M_{n}} (2d^*)^{d} & \leq \left(\frac{2C_02^{\ell(d-1)}}{b_0}\right)^{\frac{db_0 n}{2^{\ell (d-1)}}}\\
 & =\exp\left\{\frac{db_0n}{2^{\ell (d-1)}}\ln\frac{2C_0}{b_0} +\frac{db_0n\ell (d-1)}{2^{\ell (d-1)}}\ln 2\right\}\\
 & =\exp\left\{\frac{n (C_2+C_3l) }{2^{\ell (d-1)}}\right\},
\end{aligned}
\end{equation}
where $C_2 := C_2(d)$, $C_3 := C_3(d)$.

\textit{Step 4: bound the number of  possible collections of $\{d_i\}$.} The number of solutions $\{d_i\}^{M}_{i=1}$ for $d_i\in \mathbb{Z}^+ \cup \{0\}$ to $\sum_{i=1}^{M_n} d_i =N $ is $\binom{N+M_n-1}{M_n-1}$. As $\binom{N+M_n-1}{M_n-1} <\frac{(N+M_n)^{M_n}}{(M_n-1)!} < \frac{[e(N+M_n)]^{M_n}}{M_n^{M_n-1}} $, the number of solutions of \eqref{eq:bound-of-di} is bounded by 
\begin{equation}
      \begin{aligned} \label{Eq: Proposition_2_FFS_Aux_3}
\frac{1}{M_n}\sum _{N=1}^{C_0n}\left(\frac{e(N+M_n)}{M_n} \right)^{M_n} & \leq \frac{1}{M_n}\int _{M_n}^{C_0n+1+M_n} \left(\frac{ex}{M_n}\right)^{M_n} dx \le\left(\frac{e(C_0 n +M_n +1)}{M_n}\right)^{M_n+1}\\
 & \leq  \exp\left(\ln\left(\frac{2e (C_0+b_0+1)2^{l(d-1)}}{b_0}\right)\frac{2b_0 n}{2^{\ell(d-1)}}\right) \\
 &\le \exp\left\{\frac{n(C_4 +C_5l)}{2^{\ell (d-1)}}\right\},
\end{aligned}
 \end{equation}
where the third inequality follows from $M_n \ge 1$. For $l \in \mathbb{Z}^+ \cup \{0\}$, take $b_4=2d \ln(9)b_1+  \sum_{i=1}^5 C_i$ and the proposition follows. 
\end{proof}
Eventually, we derive the following proposition indicating the probability bound of $\FSIR$:
\begin{proposition}\label{prop:prob-of-FSIR}
Let $\delta > 0$ be given. Then there exists $\epsilon_2: = \epsilon_2(\delta) > 0$ such that if the maximal standard variation $\epsilon < \epsilon_2$, then for $d\ge 3$
\begin{equation}
\mathbb{P}(\eta \in \FSIR) \geq 1 -  \delta.    
\end{equation} \end{proposition}
\begin{proof}
Since we have shown from Lemma \ref{lem:subgaussian-area-variation} and the definition \ref{def:subgaussian-process} that for a fixed finite subregion $\Reg$, $\Delta_{\Reg^\prime} \tilde{F}^{k}_\Reg$ is a subgaussian process with respect to $\Reg^\prime$. Since $N(\bar{\CoC}_0(n), \bar{d}, r )$ is decreasing over $r$, we derive the following bound by Lemma \ref{lem:combinatorics-contour} and Proposition \ref{prop:entropy-control} that
\begin{equation}
\begin{aligned}
\int_0^\infty  \sqrt{\ln N(\bar{\CoC}_0(n), \bar{d}, r )}dr & \le  \nu b_3 n^{\frac{1}{2}}\sqrt{2n\ln(3^d-1)} \\
& + \sum_{\ell=1}^\infty \nu b_3n^{\frac{1}{2}}\left(2^{\frac{\ell}{2}}-2^{\frac{\ell-1}{2}}\right) \sqrt{\ln N(\bar{\CoC}_0(n), \bar{d},\nu b_3 2^{\frac{\ell-1}{2}} n^{\frac{1}{2}} )}   \\
& \le \nu b_3 n^{\frac{1}{2}}\left[ \sqrt{2n\ln(3^d-1)}  + \sum_{\ell =1 }^\infty (\sqrt{2}-1)2^{\frac{\ell-1}{2}}\sqrt{\frac{b_4\ell n}{2^{\ell(d-1)}}} \right] \\
& \le \nu b_3 b_5 n,
\end{aligned}
\end{equation}
where 
\begin{equation}
    b_5 = \sqrt{2\ln(3^d-1)}+ (1-1/\sqrt{2})\sqrt{b_4 } \sum_{\ell=1}^\infty \sqrt{\ell} 2^{-\frac{\ell(d-2)}{2}}  < \infty
\end{equation} for $d \ge 3$.
Then we can apply Theorem \ref{thm:chaining-Tail-inequality} to show that
\begin{equation}
\begin{aligned}
 & \Pb\left[\sup_{\C \in \CoC_0(n)}\left\{\sum_k \Delta_{\Int_k \C}\tilde{F}^{k}_{\Reg}  \right\}\geq \frac{\rho n}{4} \right]  \\
\leq  & \Pb\left[\sup_{\C \in \CoC_0(n)}\left\{\sum_k \Delta_{\Int_k \C}\tilde{F}^{k}_{\Reg} - \sum_k \Delta_{\emptyset}\tilde{F}^{k}_{\Reg} \right\}\geq  C \nu b_3 b_5 n + \frac{\rho n}{8}\right] \\
\\
\leq &  \Pb\left[\sup_{\C \in \CoC_0(n)}\left\{\sum_k \Delta_{\Int_k \C}\tilde{F}^{k}_{\Reg} - \sum_k \Delta_{\emptyset}\tilde{F}^{k}_{\Reg} \right\}\geq C\int_{0}^{\infty} \sqrt{\ln N(\bar{\CoC}_0(n), \bar{d},  r)}  dr + \frac{\rho n }{8}\right] \\
\leq & \exp{\left(\frac{-\rho^2n^2}{64C \mathrm{diam}(\CoC_0(n))^2 }\right)} \leq \exp{\left(-\frac{\rho^2}{128C \nu} n^{2-\frac{d}{d-1}}\right)},
\end{aligned} 
\end{equation}
The first inequality is satisfied for sufficiently small $\epsilon $ and the last inequality follows from the isoperimetric inequality \cite[ Corollary B.80]{friedli_velenik_2017}. Then the conclusion follows by summing up the inequalities for $n\ge 1$ and note that the upper bound series are convergent and the sum can be smaller than any given $\rho>0$ for sufficiently small $\epsilon$. 
\end{proof}

%\section{Results}\label{sec:results}
\section{Long-range properties of the statistical models}\label{sec:theorems}

\subsection{Proofs of the main results} 

% It is ready to show the main results 
We are now ready to present the main results of the general properties of statistical models over the lattice model. To formulate our main result rigorously, we first address the foundational question of the existence of well-defined Gibbs states in disordered systems. Building on the framework established by \cite{aizenmanRoundingEffectsQuenched1990} and \cite{newmanspatial1996}, we note that for any given realization of the disorder $\eta$, there exists a convex set of $\eta$-covariant Gibbs measures. This structure guarantees that the infinite-volume Gibbs measures transform covariantly under both lattice translations and local modifications of the disorder field, enabling the application of ergodic theory.

We now state the main result on the stability of long-range order:

\begin{theorem} \label{thm:Gibbs-measure}
Consider a statistical mechanical system on $\Zd$ ($d \geq 3$) with Hamiltonian $H_\eta$ satisfying the Peierls condition with $\Ng$ ground states and the existence of local symmetric operation. Then there exist constants $T_0 > 0$ and $\epsilon_0 > 0$ such that for all $T < T_0$ and $\epsilon < \epsilon_0$, the system admits at least $\Ng$ distinct $\eta$-covariant Gibbs measures $\{\mu^k_\eta\}_{k=1}^{\Ng}$, such that 
\begin{equation}
   \lim_{\Reg \rightarrow \Zd } \frac{1}{|\Reg|} \sum_{s \in \Reg} \mu ^k_{\eta }(\Id_{x_s = b^k})  > \frac{1}{2}.
\end{equation}
for almost all $\omega$.
\end{theorem}

\begin{proof}
The proof starts from the upper bound calculation of the single sites expectation of the statistical system over the finite region with fixed boundary condition. Consider the model of the fixed boundary condition with the boundary condition $H^{k}_{\eta, \Reg}$. Then by simple geometry we have 
\begin{equation}
     \mu^k_{\eta, \Reg }(x_0 \neq b^k ) \le \sum_{n=1}^{\infty} \sum_{\C \in \CoC_0(n), \sC \subseteq \Reg} \mu^k_{\eta, \Reg}(\C \in \TB_\Ext^{k}),
\end{equation}
where the Gibbs measure is defined as \eqref{eq:Gibbs-measure}. By the argument in Proposition \ref{prop:polymer-formula}, we have 
\begin{equation}
       \mu^{k_0}_{\eta, \Reg}(\C_0 \in \TB_\Ext^{k_0})  \le \frac{ \Xi^{k_0}_\eta(\{\C_0\})\sum_{\TB_\Ext^{k_0}, \C_0 \in \TB_\Ext^{k_0}}    \Xi^{k_0}_\eta(\TB_\Ext^{k_0} \setminus \{\C_0\})  
}  {\sum_{\TB_\Ext^{k_0}}  \Xi^{k_0}_\eta(\TB_\Ext^{k_0}) }
,
\end{equation}
where
\begin{equation}
    \Xi^{k_0}_\eta(\TB_\Ext^{k_0}):= \prod_{\C \in \TB_\Ext^{k_0}}  \exp{\left[-\frac{1}{T}D^{k_0}_\eta (\C)\right]}  \prod_k \exp\left[\frac{\Ge|\Int_k \C|+ S^{k}_{\Int_k \C}}{T}  \right]\tilde{\PF}^k_{\eta, \Int_k \C}. 
\end{equation}
Then for every contour $\C$
\begin{equation}
    \Xi^{k_0}_\eta(\{\C\}) \le \exp\left(-\frac{\rho|\sC|}{2T}\right) \prod_k \left[\frac{\Ge|\Int_k \C|+ S^{k}_{\Int_k \C}}{T}  \right]\PF^{k_0}_{\tau_{\Int_k \C}\eta,  \Int_k \C},  
\end{equation}
For $\eta \in \FSC \cap \QISC$. As a result, we have
\begin{equation}
\begin{aligned}
    \mu^{k_0}_{\eta, \Reg}(\C_0 \in \TB_\Ext^{k_0})&  \le \exp\left(-\frac{\rho|\sC|}{2T}\right)   \frac{ \Xi^{k_0}_{\tau_{\Int \C_0}\eta}(\{ b^k|_{\sC_0}\})\sum_{\TB_\Ext^{k_0}, \C_0 \in \TB_\Ext^{k_0}}    \Xi^{k_0}_\eta(\TB_\Ext^{k_0}\setminus\{\C_0\}) 
}  {\sum_{\TB_\Ext^{k_0}}  \Xi^{k_0}_\eta(\TB_\Ext^{k_0}) } \\
& \le \exp\left(-\frac{\rho|\sC|}{2T}\right) \frac{ \Xi^{ k_0}_{\tau_{\Int \C_0}\eta,\Reg}}{\Xi^{ k_0}_{\eta,\Reg}},
\end{aligned}
\end{equation}
where $\Xi^{ k_0}_{\eta,\Reg}$ is defined as in Proposition \ref{prop:polymer-formula}. Since $ \Xi^{ k_0}_{\tau_{\Int \C_0}\eta,\Reg}/\Xi^{ k_0}_{\eta,\Reg} =  \PF^{ k_0}_{\tau_{\Int \C_0}\eta,\Reg}/\PF^{ k_0}_{\eta,\Reg}$, we derive that $\mu^k_{\eta, \Reg}(\C \in \TB_\Ext^{k}) \le \exp(-\rho|\sC|/4T)$ for every $\C \in \CoC_0$ and $\eta \in \FSC \cap \QISC \cap \FSIR$. Then
\begin{equation}
     \mu^k_{\eta, \Reg }(x_0 \neq b^k ) \le \sum_{n=1}^{\infty} \exp \left[ - n \left(  
 \frac{\rho }{ 4T} -d\ln{9}-1 - \ln\Ng\right) \right],
\end{equation}
by Lemma \ref{lem:combinatorics-contour}. It is obvious that $\mu^k_{\eta, \Reg }(x_0 \neq b^k ) < 1/4$ for sufficiently small $T$ and $\eta \in \FSC \cap \QISC \cap \FSIR$, 
while by the Proposition \ref{prop:prob-of-FSC}, Corollary \ref{cor:prob-of-QISC} and Proposition \ref{prop:prob-of-FSIR}, $\Pb (\eta \in \FSC \cap \QISC \cap \FSIR) \ge 1- 3 \rho > 3/4 $ for sufficiently small $\epsilon$. Since for any fixed random field configuration $\eta$, every Gibbs measure over the whole lattice $\Zd$ is the limit point of $\mu ^k_{\eta, \Reg' }$ in the measure space, we can find $\Ng$ Gibbs measures $\mu^k_{\eta}$ such that $\Ex_\omega (\mu^k_\eta (x_0 =  b^k) ) > 1/2$. Then the conclusion follows from the construction of the $\omega$-covariant measures $\mu^k_\eta$ and the ergodicity theorem for the ergodic measure of $\omega$.
\end{proof}

\subsection{An extension over continuous models}
 
The results can also be generalized to systems with a continuous spin variable (and thus an infinite state space at each site) under certain conditions.
Suppose the spin value space $\SV$ is now a compact measurable space with the measure $\MoSV$. An extended version of the Peierls condition requires an equivalent description of the finiteness of the ground states, which inspires the definition of partition over the spin value space. A partition of the spin value space is a collection of subsets $\{\SV_i \}_{i=1}^\Ns$ such that $\SV_i \cap \SV_j = \emptyset$ for any $i \neq j$ and $\cup_i \SV_i = \SV$. Accordingly, we denote the measures restricted in every partition at site $s$ by $(\MoSV^i)_s$ such that $d(\MoSV^i)_s :=\Id_{\SV_i} d(\MoSV)_s$. We also define the coarse-grained spin configurations $\tilde{x}$ corresponding to $x$ as $\tilde{x}_s:= \SV_i$ if $ x_s \in \SV_i$ and the corresponding restricted measure as $\MoSV^{\tilde{x}_s}:= \MoSV^{i}$. Based on the setting, we then reintroduce the Hamiltonian within a finite subregion, with the restricted spin configuration $y$ acting as the boundary condition such that $\ExB \Reg  \subseteq \bar{y}$, as 
\begin{equation}\label{eq:modified-energy-function}
    H^y_{\eta, \Reg}(x) :=  \sum_{\alpha} \sum_{s \in \Reg } (h^\alpha + \eta_s^\alpha )g^\alpha_s(x^{y , \Reg}),  
\end{equation}
where 
\begin{equation}
    x^{y, \Reg}_{s} = \left\{ \begin{array}{cc}
        y_s  &  s \in  \bar{y}\\
        x_s  &  s  \in \bar{y}^c
    \end{array} \right.. 
\end{equation}
We also denote the effective Hamiltonian of a coarse-grained configuration as 
\begin{equation}\label{eq:modified-coarse-energy-function}
    H^y_{\eta, \Reg}(\tilde{x}) := - T\ln \left\{ \int \prod_{s \in \Reg} d(\MoSV^{\tilde{x}_s} )_s(x) \exp \left[ -\frac{1}{T} H^y_{\eta, \Reg }(x)\right] \right\}.
\end{equation}
Note this effective Hamiltonian depends on the temperature $T$ and we omit the dependence in the notation without ambiguity. By \eqref{eq:modified-coarse-energy-function}, we denote the effective Hamiltonian with coarse-grained configuration $\SV_i$ over $\Reg^c$ as 
\begin{equation}\label{eq:doubly-coarse-energy-function}
H^{\tilde{y}}_{\eta, \Reg}(\tilde{x}):= -T\ln \left\{  \int \prod_{s \in \ExB \Reg}d(\MoSV^{\tilde{y}_s}) (y)     \exp \left[ -\frac{1}{T} H^y_{\eta, \Reg }(\tilde{x})\right] \right\}
\end{equation}
In the following, we divide the elements of a partition into two parts: $ \{ b^k  \}_{k=1}^\Ng$ and $\{ m^k \}_{k=1}^{\Ns-\Ng}$. The first part plays the role of the ground state as in the case of finite spin values, while the second part represents the metastable states. In the following, we also write $\MoSV^k:= \MoSV^{b^k}$ for short without ambiguity. This classification is confirmed by the requirement that the only local ground states of $H_0(\tilde{x})$ for the coarse-grained spin configuration is the constant configuration $\tilde{x} \equiv b^k$. For the case of periodic coarse-grained ground states, we can also reduces it into the constant ones as shown in Section \ref{sec:notation}.    

The corresponding partition function and the corresponding Gibbs measure are 
\begin{equation}
    \PF^k_{\eta,\Reg} := \int \prod_{s \in \ExB \Reg} d(\MoSV^k )_s(y) \prod_{s \in \InB_2 \Reg }  d(\MoSV^k)_s(x) \prod_{s \in \Reg \setminus \InB_2 \Reg } d(\MoSV)_s(x) \exp \left[-\frac{1}{T} H^{y}_{\eta, \Reg }(x)\right]
\end{equation} and
\begin{equation}
    \mu^k_{\eta, \Reg} (\tilde{x}): = \frac{1}{\PF^k_{\eta, \Reg}}  \int \prod_{s \in \ExB \Reg} d(\MoSV^k)(y) \prod_{s \in \Reg}  d(\MoSV^{\tilde{x}_s} )(x) \exp\left[-\frac{1}{T}H^{y}_{\eta, \Reg}(x)\right] 
\end{equation} for any $\tilde{x}$ with $\tilde{x}_s = b^k$ for $s \in \InB_2 \Reg$. $\tilde{\PF}^k_{\eta, \Reg}$ is defined similarly. To derive the polymer formula, contours $\C$ are defined as usual except replacing the spin configuration $x$ with the corresponding coarse-grained one $\tilde{x}$ and specify the corresponding boundary set.  We also specify the exact spin configuration on $\ExB \Int \C \cup \InB \Int \C$ for the relevant quantities. We can derive the recursive formula of $\PF^k_{\eta, \Reg}$ as \eqref{eq:PF-decomposition} or \eqref{eq:recursion-pre}
\begin{equation}
\begin{aligned}
\PF^{k_0}_{\eta, \Reg}  = & \sum_{\TB^{k_0}_\Ext}  \int \prod_{s \in  \ExB \Ext \cup \InB \Ext} d(\MoSV^{k_0})_s(z)\exp\left[-\frac{1}{T} H^{z}_{\eta, \Ext}(b^{k_0}_{\Ext})\right]  \prod_{\C \in \TB^{k_0}_\Ext}  \\
 &  \prod_k \left\{ \prod_{s \in \ExB \Int_k \C \cup \InB \Int_k \C} d(\MoSV^k )_s(y)   \exp \left[-\frac{1}{T} H^{y \cup z}_{\eta , \Reg}(\C) \right]  \tilde{\PF}^y _{\eta, \Int_k \C} \right\},   
\end{aligned}
\end{equation}
where $\Ext:= \cup_{\C \in \TB^{k_0}_\Ext} \Ext \C$, and $y$ and $z$ are the restricted spin configurations over $\cup_{\C \in \TB^{k_0}_\Ext} (\ExB \Int\C \cup \InB \Int\C)$ and $\Ext$, respectively. 
Therefore, it is natural to consider the following energy difference with respect to a prescribed restricted configuration $r$ with $ r_s \in \C_s$ for $s \in \bar{r} = \InB\sC \cup \ExB \sC $ and
\begin{equation}\label{eq:modified-difference-formula}
\begin{aligned}
 D^{r,\tau}_\eta(\C) = & - T\ln\int \prod_{s \in \sC } d(\MoSV^{\C_s })_s(x) \exp\left[-\frac{1}{T}\sum_{\alpha} \sum_{s \in \sC } (h^\alpha + \eta_s^\alpha ) g^\alpha_s(x^{r, \sC}) \right] \\
& +T\ln \int \prod_{s \in \sC} d(\MoSV^k)_s(y) \exp\left[  -\frac{1}{T}\sum_{\alpha} \sum_{s \in \sC } (h^\alpha + \eta_s^\alpha )  g^\alpha_s (y^{\tau(r),\sC })\right], 
\end{aligned}
\end{equation}
where $k$ satisfies that $\C \in \CoC^k$. In this case, we also need the existence of a measurable transformation $\tau$ such that $\tau(r)_s \in b^k$ for any $s \in \bar{r}$ and the induced transformation of measures satisfies $\tau[(\mu^k)_s]=(\mu^{k'})_s$ and $r_s \in b^{k'}$. 
\begin{definition}[The extended Peierls condition]\label{def:extended-peierls-condition}
The unperturbed statistical system with the Hamiltonian $H^y_{0, \Reg}(x) $ is said to satisfy the Peierls condition if there exists a finite partition and temperature threshold $T_0>0$ such that for any $T < T_0$ and for any coarse-grained contour with respect to the partition, there exists a positive constant $\rho$ such that $  D^{r, \tau}_0(\C) \ge \rho |\sC|$ for any restricted configuration $r$ with $ r_s \in \C_s$ for $s \in \bar{r} = \InB\sC \cup \ExB \sC $ and for any measurable transformation $\tau$ such that $\tau(r)_s \in b^k$ for any $s \in \bar{r}$ and the induced transformation of measures satisfies  $\tau[(\mu^k)_s]=(\mu^{k'})_s$ and $r_s \in b^{k'} $.
\end{definition}
As a result, we can derive Proposition \ref{prop:prob-of-FSC} as usual for $D^{r, \tau}_\eta$ uniformly over the rescricted configuration $r$ given $\tau$ as in Definition \ref{def:extended-peierls-condition}. To derive the extended editions of Corollary \ref{cor:prob-of-QISC} and Proposition \ref{prop:prob-of-FSIR}, we also need an extended definition of local symmetry transformation. 
\begin{definition}[The extended local symmetry]\label{def:extended-local-symmetry}
For a statistical physical model with random fields, there exists a local symmetry with respect to the partition $ \{ b^k  \}_{k=1}^\Ng$ and $\{ m^k \}_{k=1}^{\Ns-\Ng}$ for any pair $H^{k_1}_{\eta, \Reg}$ and $H^{k_2}_{\eta, \Reg}$ if there exists a continuous configuration transformation $\bar{\tau}_{\Reg}: \SV^{\Reg\cup \ExB \Reg} \rightarrow \SV^\Reg$ and a continuous quenched parameter configuration transformation $\tau_\Reg: \QP^{\Nb \times \Reg \cup \ExB_2 \Reg} \rightarrow \QP^{\Nb \times \Reg \cup \ExB \Reg}$ for any finite region $\Reg$ such that:
\begin{enumerate}
    \item The locality, regularity and measure quasi-invariance conditions work for $\bar{\tau}_\Reg x$ and $\tau_\Reg \omega$.  
    \item  For any $\SV_i$ in the partition, $\bar{\tau}_\Reg$ always maps it into some $\SV_j$, and the induced transformation of measures satisfies $\bar{\tau}_\Reg(\MoSV^{i}) = \MoSV^{j}$ for $\MoSV^i$ and $\MoSV^j$ the corresponding restricted measures to $\SV_i$ and $\SV_j$.    
    \item The injectivity and energy quasi-invariance conditions work for the coarse-grained spin configuration $\tilde{x}$ and the corresponding effective Hamiltonian $H^{k_1}_{ \eta, \Reg}(\tilde{x}) $ and $ H^{k_2}_{\tau_{\Reg}\eta, \Reg}(\bar{\tau}_\Reg \tilde{x})$. 
\end{enumerate}
\end{definition}
It is as usual to derive the similar properties about the extended local symmetry, including Corollary \ref{cor:prob-of-QISC} and Proposition \ref{prop:prob-of-FSIR}. We also note that $\bar{\tau}_\Reg$ defined in the extended local symmetry just matches the required transformation in the definition of the extended Peierls condition. Thus we can do the same arguments as in this section to acquire the same results in Theorem \ref{thm:Gibbs-measure}. 

\section{Models}\label{sec-model-examples}

We now implement the extended Ding-Zhuang approach for characteristic cases in disordered lattice physics.

\subsection{RFIM and RFPM}

We reconsider the RFIM and RFPM on the $d$-dimensional lattice $\Zd$ in our setting, which are just the target that the original Ding-Zhuang argument focused on. The spin value space for the Ising model is $\SV = \{-1, +1\}$.
The unperturbed Hamiltonian $H_0(x)$ for the ferromagnetic Ising model is given by:
\begin{equation}
H_0(x) = -J \sum_{\langle s,t \rangle} x_s x_t,
\end{equation}
where $J > 0$ is the coupling constant, and the sum is over all unordered pairs of nearest-neighbor sites $\langle s,t \rangle$ such that $|s-t|_{L^1}=1$. 
% The nearest sites   
This can be written in the general form by identifying one type of interaction $g^{0}(x) =\frac{1}{2d} \sum_{i=1}^{2d} x_0 x_{e_i}$ for a nearest neighbor $e_i$ to the origin, with interaction strength $h^{0} = -dJ$. 
The interaction range $R_1=1$.
The random field part of the Hamiltonian couples a random field $\eta_s$ to the spin $x_s$ at each site $s$:
\begin{equation}
H'(x) = -\sum_s \eta_s x_s.
\end{equation}
This corresponds to a self-interaction term $g^{1}(x) = x_0$ (and its translates $g^{1}_s(x) = x_s$) with strength $\eta_s$. In classical setting, the random field configuration $\{\eta_s\}$ is itself a set of independent and identically distributed Gaussian random variables, so $\omega = \eta$. We can also suppose these variables to be bounded, independent,  and identically and symmetrically distributed with respect to $0$. The total Hamiltonian in a finite volume $\Reg$ with boundary condition $b$ is $H^b_{\eta,\Reg}(x)$ as given by
\begin{equation}
    H^{b}_{\eta,\Reg}(x) = -J \sum_{\langle s,t \rangle: \{s,t\}\cap\Reg  \neq \emptyset} x_s^{b,\Reg} x_t^{b,\Reg}-\sum_{s \in \Reg} \eta_s x_s^{b,\Reg}
\end{equation}

For the unperturbed system $H_0$, there are two constant local ground states: $b^{+}: =1$ and $b^{-}:= -1$. Thus, $\Ng = \Ns = 2$. To check the Peierls condition, we see that for any contour $\C$ with energy above the ground state, $\C$ comprises spins with at least one bond crossing the interface between opposite spin regions. Let $L_{\C}$ denotes interface length, the number of the pairs of opposite spins. Then energy gap is $2 J L_{\C}$. Since by the contour condition, $L_{\C} \geq |\sC\setminus \InB \sC|/2 \ge |\sC|/(2 \cdot 3^d)$.
This yields the Peierls condition with $\rho_0 = J/3^d$. To verify the local symmetry of RFIM, let $b^1 = 1$ and $b^2 = -1$.
We define the spin configuration transformation $\bar{\tau}_\Reg$ as 
\begin{equation}\label{eq:trans-RFIM1}
(\bar{\tau}_{\Reg}x)_s = 
    -x_s,
\end{equation}while the random field configuration transformation (same as the quenched configuration transformation) $\tau_\Reg$ is
\begin{equation}\label{eq:trans-RFIM2}
(\tau_{\Reg}
\eta)_s = -\eta_s. 
\end{equation}
Note that the same transformations can also be applied for $b^1 = -1$ and $b^2 = 1$. One can easily verify all the conditions of local symmetry for the above transformations. Specifically, we can obtain the strict energy and measure invariance: The transformation of $x$ does not change the interaction term given by $g^0(x)$, while the $\eta_s g^1_s(x)$ remains unchanged under the combined action of the two transformations. The random field transformation also preserves the distribution of $\{\eta^\alpha_s\}_{s\in \Reg}$.

The same argument is also for the $Q$-state Random Field Potts Model (RFPM) on the $d$-dimensional lattice $\Zd$.
The spin space is $\SV = \{1, 2, \ldots, Q\}$. The unperturbed Hamiltonian $H_0(x)$ for the ferromagnetic Potts model is given by:
\begin{equation}
H_0(x) = -J \sum_{\langle s,t \rangle} \delta(x_s, x_t),
\end{equation}
where $J > 0$ is the coupling constant, $\delta(a,b)$ is the Kronecker delta ($1$ if $a=b$, $0$ otherwise), and the sum is over all unordered pairs of nearest-neighbor sites. Suppose the random field part of the Hamiltonian is:
\begin{equation}
H'(x) = -\sum_s \sum_{q=1}^Q \eta_s^q \Id_{\{x_s=q\}},
\end{equation}
where $\Id_{\{x_s=q\}}$ is an indicator function and $\{\eta^q_s\}$ is a set of independent and identically distributed Gaussian variables (or bounded variable symmetric to $0$). Then $\Na = Q$ for the random fields and the total Hamiltonian in a finite region $\Reg$ with boundary condition $b$ is $H^b_{\eta,\Reg}(x)$ as given by
\begin{equation}
    H^b_{\eta,\Reg}(x) = -J \sum_{\substack{\langle s,t \rangle \\ \{s,t\}\cap\Reg  \neq \emptyset}} \delta(x_s^{b,\Reg}, x_t^{b,\Reg}) - \sum_{s\in \Reg} \sum_{q=1}^Q \eta_s^q \Id_{\{x_s^{b,\Reg}=q\}}.
\end{equation}

For the RFPM with $J>0$, the unperturbed system $H_0$ is the standard ferromagnetic Potts model. There are $Q$ constant local ground states sampled by $Q$ labels such that $b^q = q$ and $\Ng = \Ns = Q$. The Peierls condition can be verified as the RFIM. To verify the existence of local symmetry, we define the the spin configuration transformation $\bar{\tau}_\Reg$ between $b^i$ and $b^j$ as  
\begin{equation}
(\bar{\tau}_{\Reg}x)_s = \pi^{j-i}(x_s) , 
\end{equation}
where $\pi: \SV \to \SV$ is given by $\pi(q) = (q \pmod Q) + 1$. This is a cyclic permutation of the spin states. The corresponding random field configuration transformation can be defined as 
\begin{equation}
(\tau_{\Reg}
\eta)^q_s = \eta^{\pi^{-(j-i)}(q)}_s. 
\end{equation}
It is easy to verify the local symmetry condition as the RFIM. Here the the strict energy and measure invariance still hold. 

\subsection{Edwards-Anderson model}

The Edwards-Anderson model, or random bond Ising model, is another classical type of disordered system. A typical setting of the random bond Ising model with asymmetric distributions over bonds due to Nishimori \cite{Nishimori_1980,nishimoriinternal1981} can be represented by an Ising model (ferromagnetic or antiferromagnetic) with bond strength $\bar{J}$ coupled to a 2-valued random variables $J_{\langle s,t\rangle}$ with values $(-2\bar{J}, 0)$. In this section, we examine an extended version of the Edwards-Anderson model in which the random fields act over not only bonds but also  interacting with each spin is incorporated. Then this model can be viewed also as an extension of the RFIM. 

The Hamiltonian of the unperturbed extended model, which is unaffected by the randomness of the coupling strengths and external field, is defined as a standard Ising model:
\begin{equation}
    H_{0,\Reg}^b(x)=-\sum_{s\in\Reg}h^{J}g^J_s(x^{b, \Reg}_{s}),
\end{equation}
where $g^J_s(x)$ is given by $g^J_s(x) = \sum_{ t, |s-t|_{L^1} =1  }x_s x_t$, and the strength of the interaction $h$ is defined as  $h=\overline{J}/2$. The random fields interacting with the extended system are represented by $\eta^J_{\langle s, t \rangle}$ and $\eta^h_s$. Consequently, the Hamiltonian of the system is expressed as:
\begin{equation}
    H_{\eta,\Reg}^k(x)=-\sum_{s\in\Reg}h^{J}g^J_s(x^{b, \Reg}_{s})-\sum_{\alpha\in\{J,h\}}\sum_{s\in\Reg}\eta^{\alpha}_s g^{\alpha}_s(x^{b, \Reg}_{s})
\end{equation}
with $g^{h}_s(x)=x_s$. The random fields configuration $\{\eta^J_s\}$ and $\{\eta^h_s\}$ are supposed independent and
identically distributed. For every single $\eta^J_s$, the random variable can be a Gaussian random variable or a bounded random variable, without necessarily being symmetric to $0$. $\eta^h_s$ is designed as in the case of RFIM, where being symmetric to $0$ is necessary.    

When $\overline{J}>0$, the unperturbed system is the standard ferromagnetic Ising model without any external field. Therefore, we can acquire its ground states and derive the Peierls condition same as the case of RFIM. For the local symmetry, we can also reuse the transformations \eqref{eq:trans-RFIM1} and \eqref{eq:trans-RFIM2} for $x $ and $\eta^h$, while remain $\eta^J$ unchanged. Not only all the conditions of local symmetry are satisfied, the strict energy and measure equivalence is still maintained since the interactions given by random bonds $\eta^J$ also remain unchanged under the transformation of $x$.

The case of $\overline{J}<0$ is different, however. When $\overline{J}<0$, the unperturbed model corresponds to the antiferromagnetic Ising model without any external field. Its ground states exhibit chessboard configurations with two sublattices, where different occupations of positive and negative spins result in distinct ground states. Therefore, the period of the ground states in every coordination is $2$. Therefore, after modifying the configuration spaces, we have $\Ng = 2$ while $\Ns  = 2^d$. We can denote the two ground states by $b^e$ (even) or $b^o$ (odd) if $x_0 = +1$ or $x_0 = -1$. To show the Peierls condition, similar interfaces can be constructed by separating spins with the same orientation. Then it is sufficient to show its validity with the similar procedure.

To show that the model satisfies the local symmetry, we design a different set of transformations $\bar{\tau}_{u,\Reg}$, $\tau_{u,\Reg}^J$ and $\tau_{u,\Reg}^h$ between $b^{e}$ and $b^{o}$ such that
\begin{equation}\label{eq:trans-EA}
    \begin{aligned}
\bar{\tau}_{u,\Reg}:& (\bar{\tau}_{u,\Reg}x)_s = x_{s-u}, \\
\tau_{u,\Reg}^J:& (\tau_{u,\Reg} \eta)^J_{\langle s,t\rangle}=\eta^J_{\langle s-u,t-u \rangle }, \\
       \tau_{u,\Reg}^h:&  (\tau_{u,\Reg} \eta)^h_s=\eta^h_{s-u},\end{aligned}
\end{equation}
where $u$ represents one of the $2d$ unit vectors in the coordination directions. The transformations in the modified configuration spaces by period reduction are then induced by \eqref{eq:trans-EA}.

It is evident that locality condition and the injectivity condition for $\bar{\tau}_{u,\Reg}$, and regularity condition for $\tau_{u,\Reg}^J$ and $\tau_{u,\Reg}^h$ holds. To show the measure quasi-invariance condition, we can set $P_\Reg$ the collection of the random fields over $\tau^{-1}_u (\Reg):= \{ s| \tau_u s \in \Reg\}$ in the original lattice before modifications and the condition follows obviously. 

To show the left conditions, notice that random bonds and random fields adjacent to the boundary in the direction of translation move out of the region $\Reg$, while new random bonds and fields enter from the opposite side of the boundary due to the translation. The energy change of the transformation can be written as
\begin{equation}\label{eq:EA-inequality}
\begin{aligned}
    &H_{\eta,\Reg}^k(x)-H_{(\tau_\Reg\eta),\Reg}^{k'}(\bar{\tau}_\Reg x)\\ 
    = & \sum_{\substack{ \langle s,t \rangle \\s\in \InB \Reg\\(s+u)\notin \InB \Reg}}[\eta^J_{\langle s,t \rangle } x^{b, \Reg}_s x^{b, \Reg}_t+\eta^h_s x^{b, \Reg}_s]\\
     & -\sum_{\substack{\langle s',t' \rangle \\s'\in\InB\Reg\\(s'-u )\notin\InB\Reg}}[(\tau_{u,\Reg} \eta)^J_{\langle s',t'\rangle}( \bar{\tau}_{u,\Reg} x)^{b, \Reg}_{s'} (\bar{\tau}_{u,\Reg} x)^{b, \Reg}_{t'}+(\tau_{u,\Reg} \eta)^h_{s'}(\bar{\tau}_\Reg x)^{b^k, \Reg}_{s'}]\\
    \le & \sum_{\alpha\in\{J,h\}}\sum_{s\in\InB\Reg}|\eta^\alpha_{s}|+|(\tau_\Reg\eta)^\alpha_s|
\end{aligned}
\end{equation}
for the required spin configuration $x$ in the original configuration space. By Transforming the inequality into the modified configuration spaces, the energy quasi-invariance condition is also satisfied. The vanishing of terms in the unperturbed Hamiltonian in \eqref{eq:EA-inequality} is due to that the ground states of the unperturbed model $b^e$ and $b^o$ can be transformed to each other by one unit translation and the boundary condition $x$. 

\begin{remark}
For the classical EA model, there is no random field $\eta^h$ but only $\eta^J$ when $\overline{J}>0$. Consequently, no nontrivial random field transformation is required. This eliminates the need for the probability bound of $\FSIR$, where $d\ge 3$ is required if Proposition \ref{prop:prob-of-FSIR} is applied. Thus, long-range order results extend to $d\ge 2$ as in \cite{horiguchiexistence1982}. Note also that the $\overline{J}<0$ case is equivalent to $\overline{J}>0$ via gauge transformation over the chessboard sublattice. However, configuration space transformations induced by translation have broader applications, such as in EA models with constant self-energy $h^h x_s$, and can be applied to the hard-core model discussed below.
\end{remark}

\subsection{Quenched Fredrickson-Andersen $1$-blocked model} 

The Fredrickson-Andersen (FA) model is a classical kinetically constrained model that exhibits rich dynamical behavior. This model can be viewed as a free spin model with dynamic constraints permitting spin flips only when spins have at least the required number of positive neighbors. We study a variant named the Fredrickson-Andersen $1$-blocked model (FA-1B), a free spin model where dynamical constraints are lifted for spins with no positive neighbors. This model is equivalent to a hard-core model excluding any pair of occupied neighboring sites. The spin value space is $\SV = \{0,1\}$, where $1$ represents site occupation and $0$ vacancy. 

The FA-1B model admits a natural quenched interpretation: without dynamical constraints, the free spin model reaches an equilibrium distribution with each site occupied independently and identically distributed. When constraints are imposed, certain $\Zd$ sites become blocked: 
\begin{itemize}
    \item Occupied sites with occupied neighbors cannot change state
    \item Unoccupied sites adjacent to such occupied sites cannot change state
\end{itemize}
These form the blocked sites in the quenched FA-1B model. Consequently, the stochastic dynamics of the hard-core model are restricted to unblocked sites.

% The behavior of the free sites in this model has been investigated through computer simulations, revealing that a phase transition occurs in systems with dimension $d\geq 3$ (\textcolor{blue}{Add reference}). Additionally, analytical methods have demonstrated the nonexistence of long-range order in quenched FA-1B systems with $d=2$ (see \textcolor{blue}{Add reference here}).

To define the initial distribution of the quenched FA-1B model, we set the quenched parameter $\omega = \{0,1\}$ as the initial occupation states with the identically . Consequently, the Hamiltonian can be expressed as: 
\begin{equation}\label{eq:hard-core}
    H^b_{\eta,\Reg}(x)= \sum_{s\in\Reg}h^cg^c_s(x^{b, \Reg})-\sum_{s\in\Reg}\eta^{\Gamma,b}_s(\omega_s) g^b_s(x^{b, \Reg}),
\end{equation}
where $g^c_s(x)$ is defined as $g^c_s(x)=\mu x + \sum_{t,|s-t|_{L^1}=1}\Gamma x_s x_t$ and $h^c=\frac{1}{2}$ is the constant strength, while $g^b_s(x)$ is defined as $g^b_s(x)=\frac{1}{2}\mu x + \sum_{t,|s-t|_{L^1}=1}\Gamma x_s x_t$. The hard-core model then is the limit of \eqref{eq:hard-core} for $\Gamma \rightarrow \infty$. The random variable
$\eta^c$ is defined as a function of $\omega$ as:
\begin{equation}
\begin{aligned}
\eta^{b, \Gamma}_s = & \omega_s \left[1-\prod_{t, |t-s|_{L^1} =1}(1-\omega_t)\right] 
\\
& + (1-\omega_s)\left\{1-\prod_{t, |t-s|_{L^1} =1}\left[1-\omega_t\left(1-\prod_{r, |r-t|_{L^1} =1}(1-\omega_r)\right)\right]\right\}.
\end{aligned}
\end{equation}
It is evident that $\eta^c_s = 1$ only at sites forbidden to flip due to local constraints. As a variant of site-diluted hard-core or antiferromagnetic Ising models, the quenched FA-1B model exhibits random field configurations with weak inter-site dependencies, which is quantified by $D_c = 2$ following the notations in Section \ref{sec:notation}. These dependencies emerge from specific geometric constraints despite independently distributed quenched parameters, necessitating the full form of our arguments.

\begin{remark}
In our setup of interactions, we assume the boundedness of $g^\alpha(x)$, which is not suitable in the case of hard-core model. It suffices to suppose the interaction functions without finite upper bound if they are lower bounded, though. It is left to the readers that modification can be made so that  neither the verification of the Peierls condition nor the subgaussian estimates in Section \ref{sec:subgaussian} are affected by loss of the upper bound assumption.
\end{remark}

To verify the Peierls condition in the quenched FA-1B model. We consider the unperturbed model, the hard-core model that prevents occupied sites from being adjacent. The hard-core model, which can be seen as a limit of antiferromagnetic Ising model with a critical external field, has two ground states with chessboard configuration same as the antiferromagnetic Ising model, indicating that we are able to construct the contours $\C$ following the same procedure in the antiferromagnetic Ising model. In the hard-core model, occupied sites cannot be adjacent. Consequently, $\C$ have greater energy compared to pure ground states for at least $h^c|\sC\setminus\InB \sC|/3^d$ and the Peierls condition follows. 

To establish the local symmetry in the quenched FA-1B model, we consider the transformations $\bar{\tau}_{u,\Reg}$ and $\tau_{u,\Reg}$ as those of EA model such that
\begin{equation}
    \begin{aligned}
        \bar{\tau}_{u,\Reg}:& (\bar{\tau}_{u,\Reg}x)_s=x_{s-u}\\
        \tau_{u,\Reg}:&(\tau_{u,\Reg}\omega)_s=\omega_{s-u}
    \end{aligned}
\end{equation}
where $u$ represents one of the $2d$ unit vectors in the coordination directions and $\tau_{-u}$ is the translational operator. We explain the verification of energy quasi-invariance and measure quasi-invariance conditions in detail.
For the measure quasi-invariance, we can set $P_\Reg$ the collection of the quenched parameters over $\tau^{-1}_u (\Reg \cup \ExB_{D_c} \Reg)$ in the original lattice, larger than the case of EA model because of the weak inter-site dependencies. Transformed into the modified configuration spaces, $P_\Reg$ would satisfy the requirement and the measure quasi-invariance follows. 
For the energy quasi-invariance, we can show that the energy changes under the transformations can be bounded by 
\begin{equation}\label{eq:FA-1B}
\begin{aligned}
    \sum_{\substack{s\in\InB\Reg \\(s +  u)\notin\InB\Reg }}(1-\eta^c_s)-\sum_{\substack{t\notin\InB\Reg \\(t - u) \in\InB\Reg }}(1-\eta^c_{t-u})<\sum_{s \in \partial\Reg}|\eta_s|+|(\bar{\tau}_{u,\Reg }\eta)_s|
\end{aligned}
\end{equation}
for the required spin configurations $x$ of finite energy in the original lattice. Then the energy quasi-invariance follows by transforming \eqref{eq:FA-1B} into the modified configuration spaces. 

\subsection{Models over other lattice structures}

The generalized Ding-Zhuang argument can also be applied to other periodic lattices beyond the simple cubic lattices. To show its generality, we take the hard-core model on the following prominent three-dimensional lattice structures as examples: Body-centered cubic lattice, face-centered cubic lattice and the hexagonal close-packed lattice. 

Recall the conventions of the lattice structures before further discussion. For the body-centered cubic lattice, the conventional unit cell has lattice vectors defined as $a_1 = (1,0,0)$, $a_2 = (0,1,0)$, and $a_3 = (0,0,1)$. The atoms in the unit cell are located at positions $(0,0,0)$ and $(\frac{1}{2}, \frac{1}{2}, \frac{1}{2})$, representing the corner and body-centered atoms respectively. This structure comprises two atoms per unit cell, with the body-centered atom positioned at the geometric center of the cubic cell. The unit cell of face-centered cubic lattice has lattice vectors given also by $(1, 0, 0)$, $(0,1,0)$ and $(0, 0, 1)$. The atoms occupy positions at $(0,0,0)$, $(\frac{1}{2}, \frac{1}{2}, 0)$, $(\frac{1}{2}, 0, \frac{1}{2})$, and $(0, \frac{1}{2}, \frac{1}{2})$. This configuration includes four atoms per unit cell, with atoms located at each corner and at the centers of each face of the cubic unit cell. For the hexagonal close-packed lattice, the lattice vectors of the unit cell are defined as $ (\frac{\sqrt{3}}{2}, \frac{1}{2},0)$, $(
\frac{\sqrt{3}}{2}, -\frac{1}{2},0)$, and $(0, 0 ,\frac{2\sqrt{6}}{3})
$. The atoms are located at positions $(0,0,0)$ and $(\frac{\sqrt{3}}{3}, 0, \frac{\sqrt{6}}{3})$. This results in two atoms per unit cell, with one atom at the base of the hexagonal prism and the other at the midpoint of the prism height.

The corresponding Hamiltonian $H^b_{\eta, \Reg}$ of the hard-core model is \eqref{eq:hard-core} with the graph $G(V,E)$, where $V$ is the set of the locations of the atom while $E$ is the bonds connecting the nearest pairs of the atoms. For the unperturbed models, we can distinguish the periodic ground states as the sublattice occupations for each lattice. In the body-centered lattice, there are two sublattices, generated from $(0,0,0)$ and $(\frac{1}{2}, \frac{1}{2}, \frac{1}{2})$ by the translation of the three lattice vectors. Similarly, there are four sublattices in the face-centered lattice, generated from $(0,0,0)$, $(\frac{1}{2}, \frac{1}{2}, 0)$, $(\frac{1}{2}, 0, \frac{1}{2})$, and $(0, \frac{1}{2}, \frac{1}{2})$ by the translation of the three lattice vectors. For the hexagonal close-packed lattice, we can also generated the four sublattice, originated from $(0,0,0)$, $ (\frac{\sqrt{3}}{2}, \frac{1}{2},0)$, $(
\frac{\sqrt{3}}{2}, -\frac{1}{2},0)$ and $(1,0,0)$, translated by $(\sqrt{3}, 1 ,0)$, $(\sqrt{3}, -1,0)$ and $(\frac{-2\sqrt{3}}{3}, 0, \frac{\sqrt{6}}{3})$ respectively. 

Note also that all pairs of these sublattices can be transformed into each other by translation which shift atoms to one of its nearest neighbors. This inspires the verification of the Peierls condition of these models as in \cite{heilmannexistence1972}, which is originated from \cite{dobrushin1969_2}. The key step is to define the interfaces given $\C$ in the original lattice as in previous examples. For the systems with $K$ sublattice ground states labeled by $(1, \cdots, K)$, suppose the ground state of the outer boundary of $\C$ is of type $K$ and the interfaces are defined as the boundaries partitioning the subregions of $K+1$-type: the subregions are constructed inductively such that the subregion of type $i$ ($1 \le i \le K $) is the regions occupied by atoms in sublattice $i$ and the unoccupied sites near them and not included in the subregion of type $j$ ($j< i$). 
%The subregion of the last ($K+1$) type is all the left unoccupied sites.
The subregion of the last ($K+1$) type consists of all remaining unoccupied sites
Suppose for each site $s$, there are $P$ neighboring sites for each sublattice different from that of $s$. Then the difference $D$ of the number $N_K$ of the sites of sublattice $K$ in all subregions of type $i$ ($i \neq K$) and their occupied sites has the relation:  
\begin{equation}\label{eq:ineq-of-sublattice}
    P(N_K-D) \le P N_K - N_{K, \InB},    
\end{equation}
where $N_{K, \InB}\le N_K$ is the number of the sites of sublattice $K$ in these regions near the interfaces, which is of order $O(L_{\C})$. By replacing the configuration in these regions with the ground state configuration corresponding to sublattice $K$, the energy difference is proportional to $D$ and then also $O(L_{\C}) \sim O(|\C|)$. 

The local symmetry is designed similarly as last subsection by translations that transform one sublattice into the other. As a result, we can construct two distinct Gibbs distributions for the body-centered cubic lattice, and four distinct Gibbs distributions for the face-centered cubic lattice and the hexagonal close-packed lattice. 

\begin{remark}
The same methods can be applied to the triangular lattice in two dimension without random fields, which has three sublattice ground states. However, there exists no long-range phase by applying the Aizenman-Wehr argument for the occurrence of certain random fields. 
\end{remark}

\subsection{The continuous Ising model and the anisotropic Heisenberg model}

In Section \ref{sec:theorems}, we demonstrated that for systems coupled with random fields, long-range order occurs not only in lattice systems satisfying the Peierls condition (Definition \ref{def:Peierls-Condition}) and local symmetry (Definition \ref{def:local-symmetry}), but also in systems with continuous spin values of infinite degrees of freedom that meet an extended Peierls condition (Definition \ref{def:extended-peierls-condition}) and extended local symmetry (Definition \ref{def:extended-local-symmetry}), given a proper partition of the spin measure space. We now analyze two such continuous systems coupling with gaussian random field as examples: Ising model with continuous spin and anisotropic classical Heisenberg model.

In the continuous Ising model, the spin value space is $\SV = [-1,1]$ with $\MoSV$ the uniform distribution over the interval, and the Hamiltonian of the unperturbed model can be expressed as 
\begin{equation}
    H^y_{0,\Reg}=-\sum_{s\in\Reg}h^Jg^J_s(x^{y,\Reg})
\end{equation}
with the strength of the interaction $h$ defined as $h^J=J/2$ and $g^J_s(x)=\sum_{t,|s-t|_{L^1}=1}x_ss_t$. The random field interacting with the extended system is represented by $\eta^h$ coupling with single spin, then the Hamiltonian of the system can be denoted as
\begin{equation}
    H^y_{\eta,\Reg}=-\sum_{s\in\Reg}h^Jg^J_s(x^{y,\Reg})-\sum_{s\in\Reg}\eta^hg^h_s(x^{y,\Reg})
\end{equation}
with $g^h_s(x)=x_s$.

In order to construct contours, a partition of the spin configuration space is required. In the continuous Ising model, we adopt the following partition scheme:
\begin{itemize}
    \item Ground state $b^+$ corresponds to the spins in $(\delta,1]$.
    \item Ground state $b^-$ corresponds to the spins in $[-1,-\delta)$.
    \item Metastable state $m$ identified with spin values in $[-\delta,\delta]$.
\end{itemize}
with $0<\delta<1$ .
With the partition above, we are able to write \eqref{eq:modified-difference-formula} as
\begin{equation}
\begin{aligned}
    D^{r,\tau}_0(\C)=-T\ln\int \prod_{s \in \sC } d(\MoSV^{\C_s })_s(x)\exp\left[\frac{1}{T}\sum_{s\in\sC}h^Jg^J_s(x^{r,\sC})\right]\\+T\ln\int \prod_{s \in \sC} d(\MoSV^+)_s(y)\exp\left[\frac{1}{T}\sum_{s\in\sC}h^Jg^{J}_s(y^{\tau(r),\sC})\right],
\end{aligned}
\end{equation}
with $b^+$ as boundary condition and a fixed $r$ compatible with $b^+$. For simplicity, the temperature factor has been absorbed into the interaction strength parameter $h^J$.

To verify the extended Peierls condition (Definition \ref{def:extended-peierls-condition}), we need to prove that there exists a positive
constant $\rho$ such that $D^{r,\tau}_0(\C)\geq \rho|\sC|$ uniformly for sufficiently small $T$.
The first term of $D^{r,\tau}_0$ has a trivial bound that 
\begin{equation}
\begin{aligned}
    &-T\ln\int \prod_{s \in \sC } d(\MoSV^{\C_s })_s(x) \exp\left[\frac{1}{T}\sum_{s\in\sC}h^Jg^J_s(x^{r,\sC}) \right] \\ \geq &-\max_{x, \tilde{x} = \C}\left\{\sum_{s\in\sC}h^Jg^J_s(x^{r,\sC}) \right\} ,
\end{aligned}
\end{equation}
where $x$ ranges over the restricted configurations over $\sC$. Meanwhile, the second term of $D^{r,\tau}_0$ has a lower bound by further restricting the integration domain of the spin values:
\begin{equation}
\begin{aligned}
     & T\ln\int \prod_{s \in \sC} d(\MoSV^+)_s(y)\exp\left[\frac{1}{T}\sum_{s\in\sC}h^Jg^J_s(y^{\tau(r),\sC}) \right] \\
     \geq &  \min_{y, \tilde{y} = \bar{b}^+ } \left\{\sum_{s\in\sC}h^Jg^J_s(x^{\tau(r),\sC}) \right\} + T|\sC| \ln \MoSV(\bar{b}^+) ,
\end{aligned}
\end{equation}
where $\bar{b}^+ \subseteq b^+$ stands for a smaller area in $b^+$. We set it as $\bar{b}^+:= (1-\xi,1]$. 

We can then calculate the difference of these two terms as 
\begin{equation}
\begin{aligned}
D^{r,\tau}_0\geq & -J(2d-1+\delta)|\sC \setminus \InB \sC| \\ & +  2d J(1-\xi)^2|\sC \setminus \InB \sC| - d J\xi |\InB \sC| + T|\sC| \ln|1-\xi|.
\end{aligned}
\end{equation}
For $\xi =1-\delta /(4d 3^d)$, $0<\delta<1$ and $T$ sufficiently small, the extended Peierls condition is satisfied for the continuous Ising model.

To verify the extended local symmetry, we define the continuous configuration transformation $\bar{\tau}_{\Reg}$ as 
\begin{equation}
    (\bar{\tau}_{\Reg}x)_s=-x_s,
\end{equation}
while the random field configuration transformation (same as the quenched configuration transformation)
$\tau_\Reg$ is 
\begin{equation}
    (\tau_\Reg\eta)_s=-\eta_s.
\end{equation}
Noticing that the same strategy have been used in RFIM, one can easily verify all the conditions of local symmetry for the above transformations.

The same argument is also valid for the anisotropic classical Heisenberg model. With the notations in \cite{bortzphase1972}, the spins in the classical Heisenberg model can be described using polar coordinates with azimuthal angle $\phi$ and polar angle $\theta$. This allows us to define
\begin{equation*}
    \zeta_i=\sigma_i^z=\cos\theta_i
\end{equation*}
as the $z$ component of the spin at site $i$.
Within this framework, the Hamiltonian of the anisotropic classical Heisenberg model can be expressed as:
\begin{equation}
    H_{\eta,\Reg}^y=-\sum_{\alpha\in\{\phi,\zeta\}}\sum_{s\in\Reg}h^\alpha g^\alpha_s(x^{y,\Reg})-\sum_{s\in\Reg}\eta^h g^h_s(x^{y,\Reg}),
\end{equation}
where $h^\phi=J\gamma/2$ with $\gamma\in(0,1/2)$ and $h^\zeta=J/2$. The functions $g^\phi_s$ and $g^\zeta_s$ are defined as 
\begin{align*}
    g^\phi_s(x)=\sum_{t,|s-t|_{L^1}=1}&[(1-\zeta^2_s)(1-\zeta^2_t)]^{\frac{1}{2}}\cos(\phi_s-\phi_t).\\
    g^\zeta_s(x)&=\sum_{t,|s-t|_{L^1}=1}\zeta_s\zeta_t.
\end{align*}
The random field $\eta$ follows a normal distribution $N(0,\sigma^2)$, and the coupling function is given by $g^h_s(x)=\zeta_s$, indicating that the random field couples only with the $z$-component of the spin.

Now, we apply the partition of the spin value space as:
\begin{itemize}
    \item Ground state $b^+$ corresponds to $\zeta_s$  taking value in $(\delta,1]$.
    \item Ground state $b^-$ corresponds to $\zeta_s$  taking value in $[-1,-\delta)$.
    \item Metastable state $m$ identified with $\zeta_s$ taking value in $[-\delta,\delta]$. 
\end{itemize}
with $0<\delta<1$. 

Since our partition has no change to the value space of $\phi$, the energy variance between spins in the same coarse-grained state with different azimuthal angle would not affect the calculation $ D^{r,\tau}_0$, thus the extended Peierls condition can be verified by the same method in continuous Ising model. 

In order to verify the local symmetry, we consider similar transformations as:  
The continuous spin configuration transformation $\bar{\tau}_\Reg$:
\begin{equation}
    (\bar{\tau}_\Reg\zeta)_s=-\zeta_s\\
\end{equation}
with the polar coordinations unchanged, and the random field configuration transformation $\tau_\Reg$:
\begin{equation}
    (\tau_\Reg\eta)_s = -\eta_s.
\end{equation}
It is easy to estimate the extended local symmetry with transformations above since $\eta_s$ only coupled with $\zeta_s$.
\section{Declarations}

\bmhead{Acknowledgements}

We thank Jian Ding, Fenglin Huang, Miao Huang, Jo\~ao Maia and Aoteng Xia for beneficial discussion over this manuscript. 
The following funding supports are acknowledged: National Natural Science Foundation of China Grants No. 12247104 and No. 12047503. 
\begin{appendices}

\section{Geometric results}\label{secA1}

\begin{lemma}\label{lem:combinatorics-contour}
Denote the connected finite region in $\Zd$ ($d \ge 2$) as $\Reg \subset \Zd$ and let $\Gamma(n) = \{\Reg | |\Reg| =n, 0 \in \Reg \cup \Int \Reg \}$ be the collection of connected subregions with a fixed number $n$ of the points including the origin. Then we have a upper bound of $|\Gamma(n)|$ as 
\begin{equation}
    |\Gamma(n)| \le e^n (3^d-1)^{2n}.
\end{equation}
\end{lemma}
\begin{proof}
For a connected region of $\Reg$ such that $0 \in \Reg$, we can construct a connected graph $(G,E)$, where $G$ all the points in $\Reg$ and $E$ the edge $(u,v)$ for every pair of $u,v \in \Reg$ with $|u-v|=1$. It is a basic fact that there is a spanning tree $T$ of the connected graph $G$, and an injective mapping from the set of connected $T$ with $n$ edges containing the origin to the set of paths of length $2n$ starting from the origin can also be established by an algorithmic construction of the spanning tree. To construct such a tree, we perform a depth-first search (DFS) traversal of $G$ starting from the origin.  For each edge encountered during the DFS with untouched points in $G$, we append the directed edge, label the point to be touched and search all the other edges of that point. This procedure generates a path $P$ of length exactly $2n$ that encodes the original subgraph $G$. The injectivity of this mapping follows from the fact that different subgraphs will produce different DFS traversals, and thus different paths, while different graph $G$ have also different spanning trees. Therefore, The number of such $\Reg$ is upper bounded by the total number of possible paths of length $2n$ starting from the origin, which is at most $(3^d-1)^{2n}$. 

For a connected region $\Reg$ such that $0 \in \Int \Reg$, it must intersect at some point in the positive direction of the first coordination, which has no more than $n$ possible choices. As a result, $|\Gamma(n)|\le (n+1)(3^d-1)^{2n} \le e^n (3^d-1)^{2n}$.  
\end{proof}

\section{Variational results}
\begin{lemma}\label{lem:opt-3point-generation}
Consider the following optimization problem:
\begin{equation}
\max_{x_i, p_i \in \mathbb{R}} \sum_{i=1}^3 p_i e^{\lambda x_i}  
\end{equation}
with the constraints of moments:
\begin{equation}\label{eq:constraint-moment}
    \sum_{i=1}^3 p_i = 1, \quad \sum_{i=1}^3 p_i x_i = 0, \quad \sum_{i=1}^3 p_i x_i^2 = \sigma^2
\end{equation}
and the constraints of bounds $
     p_i \ge 0, \quad  - 1 \le x_i \le 1 $. Then for $\sigma \le 1$, there exists the unique maximizer of the optimization problem up to a symmetric transformation such that 
\begin{equation}
    \begin{array}{cc}
        x_1= -\sigma^2, & x_2 =1,\\
        p_1 = \frac{1}{1+\sigma^2}, & p_2 = \frac{\sigma^2}{1+\sigma^2},
    \end{array}
\end{equation}
while $p_3=0$ and $x_3$ is any point of $[-1,1]$. The corresponding maximal value is $\frac{e^{-\lambda \sigma^2} + \sigma^2 e^{\lambda}}{1 + \sigma^2}.$
\end{lemma}
\begin{proof}
Since this is a convex optimization problem, we can use the Karush-Kuhn-Tucker (KKT) method to solve it, and the corresponding Lagrangian is:
\begin{equation}
\begin{aligned}
\mathcal{L} & =\sum_{i=1}^3 p_i e^{\lambda x_i} - \alpha \left( \sum_{i=1}^3 p_i - 1 \right) - \beta \left( \sum_{i=1}^3 p_i x_i \right) - \gamma \left( \sum_{i=1}^3 p_i x_i^2 - \sigma^2 \right)  \\
& + \sum_{i=1}^3 \mu_i (x_i + 1) + \sum_{i=1}^3 \nu_i (1 - x_i) + \sum_{i=1}^3 \xi_i p_i,    
\end{aligned}
\end{equation}
where $\mu_i \ge 0$, $\nu_i \ge 0$, $\xi_i \ge 0$ and $\mu_i>0$ ($\nu_i > 0$, $\xi_i>0 $) enforce $x_i = -1$ ($x_i = 1$, $p_i = 0$). The conditions yield
\begin{equation} \label{eq:var-of-p}
   \frac{\partial \mathcal{L}}{\partial p_k} = e^{\lambda x_k} - \alpha - \beta x_k - \gamma x_k^2 = 0. 
\end{equation}
for any $p_k >0$. Thus, all support points lie on the intersection of quadratic $q(x) = \alpha + \beta x + \gamma x^2$ with the exponential function $e^{\lambda x}$ for a nontrivial three-point distribution. In addition, for $x_k \in (-1, 1)$ (internal points), 
\begin{equation}\label{eq:var-of-x}
     \frac{\partial \mathcal{L}}{\partial x_k} = p_k \lambda e^{\lambda x_k} - p_k \beta - 2p_k \gamma x_k = 0   \end{equation}
We prove the lemma by systematically eliminating asymmetric configurations.

\textit{Eliminate the cases of more than $2$ internal points.} Suppose  $x_1, x_2 \in (-1,1)$ with $p_1, p_2>0 $ and $x_1 < x_2$. Points $x_1, x_2$ satisfy \eqref{eq:var-of-p} and \eqref{eq:var-of-x}. Define $h(x) = e^{\lambda x} - q(x)$. Then $h(x_k) = 0$ and $h'(x_k) = 0$ for $k=1,2$. By the mean value theorem, there exists a point $\bar{x}\in (x_1,x_2)$ such that $h'(\bar{x})=0$. Thus, $h'(x)$ has at least three roots and accordingly at least one root for $h'''(x)$. However, $h'''(x) = \lambda^3 e^{\lambda x}>0$, which leads to a contradiction. 

\textit{Eliminate the case with the nontrivial three-point distribution.} Assume $x_1 = -1$, $x_3 = 1$, $x_2 = c \in (-1, 1)$ and $p_i > 0$ for all $i$. By \eqref{eq:var-of-p}:
\begin{equation}
e^{-\lambda} = \alpha - \beta + \gamma, \quad e^{\lambda} = \alpha + \beta + \gamma, \quad e^{\lambda c} = \alpha + \beta c + \gamma c^2.
\end{equation}
While by \eqref{eq:var-of-x} for $x_2 = c$:
\begin{equation}
\lambda e^{\lambda c} = \beta + 2\gamma c.
\end{equation}
As in the discussion above, it is easy to check that there exist $\bar{x}_1 \in (-1,c)$, and $\bar{x}_2 \in (c, 1)$ such that $h'(\bar{x}_k)=0$ for $k=1,2$. Then there exists at least three roots for $h'(\bar{x}_k)=0$ and the contradiction follows. 

\textit{Calculate the case of $p_3= 0$.} Assume $p_1 > 0$, $p_2 > 0$, $p_3 = 0$ (without loss of generality). The distribution reduces to two points, say $a, b \in [-1, 1]$, with probabilities $p$ and $1-p$. $p$ and $ab$ are solved by \eqref{eq:constraint-moment} as $p = \frac{b}{b-a}$ and $\sigma^2 = -a b$. Assume $b > a$. Then the optimization problem is reduced into the maximization of the following formula:
\begin{equation}
G(a) = \frac{\sigma^2 e^{\lambda a} + a^2 e^{-\lambda \sigma^2 / a}}{\sigma^2 + a^2}, \quad a \in [-1, -\sigma^2].
\end{equation}
Then the maximum occurs at $a = -\sigma^2$ verified via derivative analysis. Then the conclusion follows.
\end{proof}
\begin{lemma}\label{lem:variation-boundedness}
Let $X \in \mathbb{R}$ be a random variable such that $\mathrm{Var}(X):=\Ex( X - \Ex X)^2 = 1$. Let $F(x)$ be a Lipschitz function with $|F(x)-F(y)| \le |x-y|$ for all $x,y$. Then the variation has a sharp bound
\begin{equation}
    \mathrm{Var}[F(X)] =\Ex\left[( F(X) - \Ex F(X))^2\right] \le  1.
\end{equation}
\end{lemma}
\begin{proof}
Define $G(x) = F(x) - F(\Ex X)$. Then $\text{Var}(G(X)) = \text{Var}(F(X))$ since variance is translation-invariant. By the Lipschitz condition, for any $x \in \mathbb{R}$ 
\begin{equation}|G(x)| = |F(x) - F(\Ex X)| \leq |x-\Ex X|,\end{equation}
so $G(x)^2 \leq x^2$. Taking expectations:  
\begin{equation}\Ex\left[G(X)^2\right] \leq \mathrm{Var}(X) = 1.\end{equation}
Thus,  
\begin{equation}\text{Var}(F(X)) = \Ex\left[G(X)^2\right] - \left(\Ex [G(X)]\right)^2 \leq \Ex\left[G(X)^2\right] \leq 1.\end{equation}
Since equality holds when $F(x) = x$, the conclusion follows. 
\end{proof}

\end{appendices}

% \nocite{*} 
% \bibliographystyle{plain}  % 选择样式
\bibliography{refDingZhuang}  


\end{document}